\documentclass [10pt] {article}
\usepackage{geometry}
\usepackage[english]{babel}
\usepackage[utf8x]{inputenc}
\usepackage{amsfonts}
\usepackage{mathtools}
\usepackage{amsmath}
\usepackage{amsthm}
\usepackage{pict2e}
\usepackage{xcolor}

\theoremstyle{plain}
\newtheorem{thm}{Theorem}[section]
\newtheorem{cor}[thm]{Corollary}
\newtheorem{lem}[thm]{Lemma}

\theoremstyle{definition}
\newtheorem{defn}{Definition}

\theoremstyle{remark}
\newtheorem{oss}{Observation}

\title{Fake dark matter from retarded distortions}
\author{Federico Re$^*$}

\numberwithin{equation}{section}
\begin{document}

\maketitle

$^*$DiSAT, Insubria University, Via Valleggio 11, Como, Italy; and INFN, via Celoria 16,
20133, Milano, Italy. fre@uninsubria.it

\begin{center}
	\textbf{Abstract}
\end{center}
We push ahead the idea developed in \cite{re}, that some fraction of the dark matter and the
dark energy can be explained as a relativistic effect. The inhomogeneity matter generates
gravitational distortions, which are general relativistically retarded. These combine in a
magnification effect since the past matter density, which generated the distortion we feel
now, is greater than the present one. The non negligible effect on the averaged expansion
of the universe contributes both to the estimations of the dark matter and to the dark
energy, so that the parameters of the Cosmological Standard Model need some corrections.

In this second work we apply the previously developed framework to relativistic models
of the universe. It results that one parameter remain free, so that more solutions are
possible, as function of inhomogeneity. One of these fully explains the dark energy, but
requires more dark matter than the Cosmological Standard Model ($91\%$ of the total matter).
Another solution fully explains the dark matter, but requires more dark energy than
the Cosmological Standard Model ($15\%$ more). A third noteworthy solution explains a
consistent part of the dark matter (it would be $63\%$ of the total matter) and also some of
the dark energy ($4\%$).

\section{Introduction}

The evidences of dark matter are of many kinds. We can roughly divide the \lq\lq dark matter
phenomena\rq\rq in two categories: global dark matter effects, which consist in unexpected
values of some cosmological parameters (the deceleration parameter \cite{perlmutter}, \cite{riess}, the deuterium
abundance, the power spectrum of CMB anisotropies, etc.), and local dark matter
effects, appearing in observations of astronomical objects (the galaxies rotation curves \cite{rotation curve},
the virial of galaxy clusters \cite{clusters}, gravitational lensing \cite{g lensing}, etc.). In both cases are exclusively
gravitational phenomena, anomalies of gravitation with respect to what we expect:
some of them derive directly from the study of a gravitationally interacting system, as a
galaxy, or from the space-time dynamics, as the deceleration parameter; for some others
the derivation is indirect, like for the deuterium abundance, which, anyway, depends on
\lq\lq dark matter contributions\rq\rq, i.e. on gravitation effects that cannot be associated to a
distribution of baryonic matter. Several followed hypotheses on dark matter assume that
it could also interact weakly, at least, but such \lq\lq weakly dark matter effects\rq\rq have not been
observed ut to now \cite{Craig:2013cxa}, \cite{HDM}.

The evidences of dark energy are restricted to more specifi phenomena. There are
only global dark energy effects, which consist in unexpected values of some cosmological
parameters, due to a negative deceleration parameter\cite{perlmutter}, \cite{riess}. Again, it is a purely
gravitational phenomenon, consisting in an unexpected distortion of the space-time metric.
Dark energy is often described as an exotic kind of energy having a negative pressure, or
as a \lq\lq vacuum energy\rq\rq; these assume some non gravitational interaction. However, none of
such new non gravitational interaction have been observed yet, and as for the dark matter,
all our knowledge about the dark energy comes just from gravitational phenomena.

In the last twenty years, several lines of research have been opened that seek to study
some essentially relativistic effects in cosmology. MOND theories are an attempt in this
direction, but they actually didn't give rise to a good matching with data \cite{Aguirre:2001xs}, \cite{mond1}, \cite{mond2}.
Rather, we will consider rather the attempts to obtain unusual effects from the usual General
Relativity.

A general relativistic explanation of the unexpected distortion of the space-time tensor,
even without the presence of real energy-matter, would provide the justification of at least
a part of dark matter and/or dark energy effects. A point is that the formalism usually
used to show these effects is not truly general relativistic. Especially for local effects, it is
adopted the newtonian approximation, while for global ones, one assumes a friedmannian
model for the expansion, with matter and curvature assumed to be homogeneous, which
could be an oversimplification. By now, we will refer to the newtonian approximation, for
galaxies, and the homogeneous friedmaniann model, for the universe expansion, as \lq\lq the
standard calculations\rq\rq.

There are three principal lines of research that try to explain (at least a fraction of)
dark matter and/or dark energy as a difference between the standard calculations and
a more precise general relativistic model. One of them started form a deep analysis of
the coordinate systems adopted in observational general relativity \cite{Alba:2006ea}, \cite{Lusanna:2006wx}, \cite{Lusanna:2009fd}, \cite{Lusanna:2010pb}, \cite{Lusanna:2011rm},
\cite{Lusanna:2012kx}. Recently it has had some confirmations by observations on the rotation curves of our
galaxy \cite{Balasin:2006cg}, \cite{Crosta:2018var}. Another line of research explores the backreaction effects, which means the
discrepancy between the spatial curvature due to an averaged quantity of matter, and the
averaged spatial curvature due to an inhomogeneous quantity of matter \cite{Buchert:2007ik}, \cite{Vigneron:2019dpj}; it doesn't
vanish in general since the Einstein Equations are not linear. This effects needs averaging
on large dominions, and arise the standard deviation of the spatial curvature.

The third line considers two main ingredients: the inhomogeneity of the matter
distribution, and the retard of the space-time distortions that it generates. In the universe,
the matter is inhomogeneous and anisotropic, so that the Birkhoff Theorem cannot be
applied in general; the distortions due to the inhomogeneity do not cancel out. The
Cosmological Principle states that, at large scale, the matter distribution is near to
a homogeneous one, so that we must remember that globally the matter inhomogeneities
are quite small. However, the second ingredient, the retarded distortions, provides a
magnification effect, for which the resulting space-time distortion is not negligible even
if its source is small. This magnification is due to the fact that the distortions we feel
nowadays were generated by matter sources in the past, when the matter density was
greater, because of universe expansion. The actual distortion is a superposition of
all retarded distortions from all past times, which predictably have a singularity at the
Big Bang time. However, there is also a decrease with the distance. It needs a precise
calculation to see if it prevails the magnification or the reduction, and \cite{sergio}, \cite{re} highlighted
that the magnification dominates.

These ideas were presented \emph{in nuce} in the article \cite{sergio}, which however had several lacks.
It doesn't work in a general relativistic context. The authors adopted a linearized gravity
model on a minkowskian background, where the simulated expansion of the universe is
forced by hand, and an effective FRWL metric is imposed with a \lq\lq compatibility condition\rq\rq.
In \cite{re} we firstly developed a truly general relativistic framework apt for this line of investigation.
Since the matter inhomogeneity is small according to the Cosmological Principle,
we considered it as a first order perturbation on a homogeneous FRWL background metric.
We derived the linearized Einstein Equations, which returns the metric distortion, at the
first order. Choosing the suitable gauge, these result to be wave equations, so that the
metric distortion is retarded, traveling at the speed of light. This framework have then
been tested on a specific universe, which was not a realistic one, but a model chosen such
that the linearized Einstein Equations have constant coefficients, with the only aim to
make them easier to be solved. We obtained an explicit formula for the retarded gravitational
potentials. Averaging the resulting metric, we got an evaluation of the global effects
of such perturbative correction, which manifested the magnification effect we hoped, and
was able to explain a non negligible fraction of the dark matter and the dark energy.

In the present article, we will push forward such treatment in a substantial way, by
applying the method to a general realistic background universe, with any possible combination
of matter and energy components. The main consequence is that the linearized
Einstein Equations will not have constant coefficients, in general. From those, we will
obtain an averaged metric, which can be compared with the metric assumed by the Cosmological
Standard Model. We will obtain a set of consistence conditions, which reduce
the acceptability of the possible universe. For acceptable universes, the gap between our
averaged metric and the standard calculation provides some quantity of \lq\lq fictitious\rq\rq matter
and dark energy. We will apply these formulas to the case of the real universe, which
depends in general on three parameters: the real quantity of radiation, of matter, and
of dark energy (where the last one can eventually be zero, if all the dark energy can be
explained as fictitious). Substituting the most recent cosmological data, two parameters
will be fixed. The last one will leave a range of possible solutions, some of which looks of
particular interest.

The paper is structured as follows. In \S2 we outline the perturbative method, defining
the background universe and the perturbed universe, and showing the the fictitious matter
and dark energy from the gap between them. In \S3 we recall the linearized Einstein
Equations from \cite{re}. In \S4 we average the metric perturbations, ignoring temporarily
their precise development, and we get a formula for the fictitious matter and dark energy.
In \S5 we reduce the PDEs for the evolution of the metric perturbations to ODEs, with an
averaging procedure. In \S6 we solve such ODEs with a single-component approximation,
finding also some \lq\lq Selfconsistence Conditions\rq\rq for the components that can fill the universe.
In \S7 we apply all the previous formulas to our universe, getting an explicit model,
which will result to leave a free parameter. The model is numerically solved in \S8, finding a set
of solutions depending on the free parameter.

\section{Framework and notations}

We will adopt the following terminology. With \lq\lq radiation\rq\rq we mean any component of
the universe with a pressure $p=w\rho$ with $w=\frac{1}{3}$. We call \lq\lq matter\rq\rq any component with
$w=0$; and \lq\lq dark energy\rq\rq any one with $w=-1$. A component with $w<0$ will be called
\lq\lq exotic\rq\rq.

Moreover, we call \lq\lq total matter\rq\rq the quantity of matter $\Omega_{M0}$, which is required in the
Cosmological Standard Model in order to explain the observed deceleration parameter;
analogous for the \lq\lq total dark energy\rq\rq $\Omega_{\Lambda0}$. We call \lq\lq dark matter\rq\rq $\Omega_{DM0}$ the part of total
matter that is not observed, unless indirectly via gravitational phenomena. The observed
part is essentially the \lq\lq baryonic matter\rq\rq $\Omega_{BM0}\cong\Omega_{M0}-\Omega_{DM0}$.

\subsection{Perturbative method}

Let us consider a universe filled with any choice of components, each one characterised by
its constant $w$. $\rho=\sum_w\rho_w$ is the true matter-energy density. Only the matter component
$\rho_M(\underline{x};\tau)$ can be inhomogeneous. Let us call its homogeneous part
\footnote{We can interpret it as the mean density out of the galaxies. It is very low, but it is not zero, since
even between the galaxies we don't have the perfect void. However, a universe can also have $\bar{\rho}\equiv0$; this
would mean that all its matter is inhomogeneous.}%
\begin{equation}
	\bar{\rho}(\tau):=\min_{\underline{x}}\rho(\underline{x};\tau).
\end{equation}

Let now consider a background universe, approximating our true universe at zero
order, so assuming it filled with a perfectly homogeneous $\bar{\rho}$. We assume this universe to
be spatially flat, in order to keep all calculations as simple as possible. For the same reason,
we assume irrotational matter. Then, for a more realistic description of the universe, we
perform a first order perturbative. The perturbation of the energy-matter, which in fact
consists only on matter, is%
\begin{equation}
	\tilde{\rho}(\underline{x};\tau):=\rho(\underline{x};\tau)-\bar{\rho}(\tau)=\rho_M(\underline{x};\tau)-\bar{\rho}_M(\tau).
\end{equation}

Notice that this inhomogeneity is always non negative. In a symmetric way, we could also define%
\begin{equation}
	\bar{\rho}(\tau):=\max_{\underline{x}}\rho(\underline{x};\tau);
\end{equation}%
in such a case, $\tilde{\rho}$ would be non positive. It is a matter of convenience to fix the choice
with always positive or negative $\tilde{\rho}$.



\subsection{Background evolution}

We adopt the most minus signature and natural units, so that $c=1$. The background
metric is%
\begin{equation}
	\bar{g}_{\mu\nu}(\tau)=a(\tau)^2(d\tau\otimes d\tau-\delta_{ij}dx_idx_j).
\end{equation}

Any quantity has associated an \lq\lq unperturbed\rq\rq (or \lq\lq averaged\rq\rq) version $\bar{Q}$, computed using
the background metric $\bar{g}_{\mu\nu}$, and a \lq\lq perturbed\rq\rq (or \lq\lq true\rq\rq) version $Q$ computed using
the real metric $g_{\mu\nu}$. Its \lq\lq perturbation\rq\rq is the difference $\tilde{Q}:=Q-\bar{Q}$, which we consider
negligible beyond the first order.

$\tau$ is the conformal time for the unperturbed metric. The dot will always denote derivation
with respect to the conformal time: $\dot{Q}:=\partial_{\tau}Q$. $a(\tau)$ is the unperturbed expansion
parameter, so that $\bar{t}=\int a(\tau)d\tau$ is the usual (unperturbed) time. $H(\tau):=\frac{\dot{a}}{a}$ is the Hubble
parameter for the unperturbed model. $\tau$, $a$ and $H$ are written without the overline, with
abuse of notation, for better readability. Their perturbed versions will be $\bold{a}$ and $\bold{H}$.

$\tau_0$ is the actual instant for the unperturbed universe, s.t. $a(\tau_0)=1$. $t_0$ is the present for
the true model, s.t. $\bold{a}(t_0)=1$. The \lq\lq0\rq\rq label means evaluation in the present time, both for
unperturbed and perturbed quantities. Notice that in general $t_0\neq t(\tau_0)$, again with a
little abuse of notation.

$a(\tau)$ solves the Friedmann Equations%
\begin{equation*}
	\begin{cases}
		3H^2=8\pi Ga^2\bar{\rho} \\

		\dot{\bar{\rho}}=-3H(\bar{\rho}+p)
	\end{cases},
\end{equation*}%
with $p(\tau)=\sum_wp_w(\tau)=\sum_ww\bar{\rho}_w(\tau)$. From the second Friedmann Equation we know%
\begin{equation*}
	\bar{\rho}_w(\tau)=\bar{\rho}_{w0}a(\tau)^{-3(1+w)}.
\end{equation*}

We can describe the background components as%
\begin{equation}
\label{bg comp}
	\bar{\Omega}_w(\tau):=a(\tau)^2\frac{\bar{\rho}_w(\tau)}{\bar{\rho}_0}=\bar{\Omega}_{w0}a(\tau)^{-1-3w}
\end{equation}%
s.t. $\bar{\rho}_0=\frac{3H_0^2}{8\pi G}$, so $\sum_w\bar{\Omega}_{w0}=1$. The first Friedmann Equation says now that $H(\tau)^2=H_0^2\sum\bar{\Omega}_w(\tau)$,
from which we get a ODE for the evolution of $a$%
\begin{equation}
\label{bg evol}
	\left(\frac{\dot{a}}{H_0}\right)^2=\sum_w\bar{\Omega}_{w0}a^{1-3w}.
\end{equation}

We define $a(\tau)$ as a maximal solution of this ODE, with maximal domain $(\tau_I; \tau_F)$, eventually
unbounded. The initial condition for $a$ is provided by the request that $\lim_{\tau\rightarrow\tau_I}a(\tau)=0$.
The radius of the visible universe results to be $R(\tau)=\tau-\tau_I$; notice that it is always infinite if $\tau_I=-\infty$.

\subsection{Comparison with the Cosmological Standard Model}

Averaging it $\tilde{\rho}$ on the whole space, we get $\langle\tilde{\rho}\rangle(\tau)$. The average on the whole $\underline{x}\in\mathbb{R}^3$ is defined
as the average on some increasing sequence of compact domains $\mathcal{D}\subset\mathbb{R}^3$ filling the whole
space in the limit%
\begin{equation*}
	\langle Q\rangle:=\lim_{\mathcal{D}\rightarrow\mathbb{R}^3}\langle Q\rangle_\mathcal{D}:=\lim_{h\rightarrow\infty}\frac{1}{h^3|\mathcal{D}_0|}\int_{\frac{\underline{x}}{h}\in\mathcal{D}_0}Q(\underline{x})d^3\underline{x}, \forall\mathcal{D}_0\subset\subset\mathbb{R}^3.
\end{equation*}

We can define the \lq\lq inhomogeneous matter\rq\rq as a part of the total matter component%
\begin{equation}
	\Omega_{IM}(\tau):=\frac{\langle\tilde{\rho}\rangle(\tau)}{\rho_0}=\frac{8\pi G}{3\bold{H}_0^2}\langle\tilde{\rho}\rangle(\tau).
\end{equation}

The Cosmological Standard Model measures the cosmic components via the observed
deceleration parameter%
\begin{equation}
\label{q0 sys}
	\begin{cases}
		\sum_w\Omega_{w0}=1 \\

		\frac{1}{2}\sum_w(1+3w)\Omega_{w0}=q_0:=-\frac{\partial_t^2\bold{a}_0}{\bold{H}_0^2}
	\end{cases},
\end{equation}%
but it assumes a homogeneous $\rho$. Since it is not our case, $\bold{a}$ will be obtained by just an
adaptation of the true space-time metric to a FRWL one%
\begin{equation}
	\langle g_{\mu\nu}\rangle:=dt\otimes dt-\bold{a}(t)\delta_{ij}dx_idx_j.
\end{equation}

This provides a distortion of the expansion law, so that in general $\bold{a}(t)\neq a(\bar{t})$, $q_0\neq-\frac{\ddot{a}}{H_0^2}$,
and $\Omega_{w0}\neq\bar{\Omega}_{w0}$. To fit the two conditions (\ref{q0 sys}), two more parameters are needed.
Interpreting the distortion as the unexpected presence of matter and dark energy, we will
see an effect of \lq\lq fictitious matter and dark energy\rq\rq. We evaluate them as $\Omega_{FM0}$, $\Omega_{F\Lambda0}$.

\begin{oss}
	Since they come from a global evaluation and we used a first order approximation, these
fictitious components will result to be proportional to the total perturbation
$\Omega_{IM0}$. Thus, these global effects do not depend on the spatial distribution of the matter
inhomogeneities, but only on their total amount.
\end{oss}

The matter and the dark energy components used in the CSM have a \lq\lq true\rq\rq and a fictitious
part%
\begin{equation}
	\Omega_{M0}:=\Omega_{TM0}+\Omega_{FM0}, \; \; \; \Omega_{\Lambda0}:=\Omega_{T\Lambda0}+\Omega_{F\Lambda0}.
\end{equation}

The other components are all \lq\lq true\rq\rq. The \lq\lq true\rq\rq parts must be proportional to the same
components of the background universe%
\begin{equation}
	\Omega_{Tw0}:=\frac{\bar{\Omega}_{w0}}{\sum_{w'}\Omega_{Tw'0}}=\frac{\bar{\Omega}_{w0}}{1-\Omega_{FM0}-\Omega_{F\Lambda0}},
\end{equation}%
with the exception of matter, for which we have to add again the inhomogeneous part%
\begin{equation}
	\Omega_{TM0}=\frac{\bar{\Omega}_{M0}}{1-\Omega_{FM0}-\Omega_{F\Lambda0}}+\Omega_{IM0}:=\Omega_{HM0}+\Omega_{IM0}:=\Omega_{BM0}+\Omega_{TDM0}.
\end{equation}

Remember $\Omega_{IM0}$ can be considered as positive or negative. In the second case, the homogeneous
approximation $\bar{\rho}$ is a rounding up, so that $\Omega_{HM0}>\Omega_{TM0}$.

Some of the true matter must be the baryonic matter we know to exists. If there is still
some part left, it is \lq\lq true dark matter\rq\rq $\Omega_{TDM0}$. It is some kind of matter that actually
exists, which gravitational action is not just a relativistic effect, but that is not a directly
observable matter, like primordial black holes, neutron stars, and so on and so forth.

\subsection{Classifying possible results}

We can apply this framework to a universe filled by any choice of components $\{w\}$.

\begin{defn}
	We will call \lq\lq selfconsistent\rq\rq a choice for which
	\begin{itemize}
		\item all calculations return a finite result;

		\item the linearized Einstein Equations give a unique solution;

		\item the perturbative method is justified by small enough perturbations.
	\end{itemize}
\end{defn}
We can write the last condition as
\begin{equation*}
	\begin{cases}
		|\tilde{g}_{\mu\nu}|\ll|g_{\mu\nu}| \\

		|\Omega_{IM0}|\ll\Omega_{TM0}
	\end{cases},
\end{equation*}
where the second condition means that the Cosmological Principle holds.

\begin{defn}
	We will call \lq\lq acceptable\rq\rq a choice for which
	\begin{equation*}
		\begin{cases}
			\forall w: 0\leq\Omega_{Tw0}\leq1 \\

			\Omega_{TDM0}\leq0
		\end{cases},
	\end{equation*}
	The second part states that all the baryonic matter we see is really existing, so it is included in the model.
\end{defn}
\begin{oss}
	Notice that the fictitious components can be negative, and such a case
means that the dark matter and/or the dark energy is not explained at all, but rather its
quantity is more than what is predicted by the CSM.
\end{oss}

\begin{defn}
	We will call \lq\lq good\rq\rq the choices for which both the dark matter and the
dark energy are explained, at least for some fraction, i.e.
	\begin{equation*}
		\begin{cases}
			\Omega_{TDM0}<\Omega_{DM0} \\

			\Omega_{T\Lambda0}<\Omega_{\Lambda0}
		\end{cases}.
	\end{equation*}

	Even better choices are whose which fully explain the dark matter and/or the dark energy, i.e. $\Omega_{TDM0}=0$, $\Omega_{T\Lambda0}=0$.
\end{defn}

\section{Linearized Einstein Equations}

Let us summarize one of the main results of \cite{re}.

\begin{thm}
	The linearized Einstein Equations admit at first order a solution%
	\begin{equation}
		g_{\mu\nu}(\underline{x};\tau)=a(\tau)^2\left(\begin{matrix}
			2A+1 & -\vec{\nabla}B \\

			-\vec{\nabla}B & (2C-1)\delta_{ij}
		\end{matrix}\right),
	\end{equation}%
	in the harmonic gauge%
	\begin{equation}
	\label{gauge}
		\begin{cases}
			\dot{A}+4HA+\nabla^2B+3\dot{C}=0 \\

			A+\dot{B}+2HB=C
		\end{cases},
	\end{equation}%
	where the \lq\lq metric perturbations\rq\rq $A, B, C$ follow the PDEs%
	\begin{equation}
	\label{lin ein eq}
		\begin{cases}
			\Box A -2H\dot{A}+2(\dot{H}-2H^2)A=4\pi Ga^2\tilde{\rho} \\
			\Box B -2H\dot{B}+2(\dot{H}-2H^2)B=16\pi Ga^2q \\
			\Box C -2H\dot{C}-2\dot{H}A=4\pi Ga^2\tilde{\rho}
		\end{cases}.
	\end{equation}
\end{thm}

The box operator denotes the flat d'alembertian%
\begin{equation*}
	\Box:=\delta_{ij}\partial_i\partial_j-\partial_{\tau}^2.
\end{equation*}

The source $q(\underline{x};\tau)$ has nothing to do with the deceleration parameter $q_0$, but it is an
expression for the (irrotational) velocity field%
\begin{equation*}
	\vec{\nabla}q:=\vec{q}:=(\bar{\rho}+p)\vec{v}.
\end{equation*}

A general solution of the linearized Einstein Equations has also a wave term%
\begin{equation*}
	\Box\tilde{g}_{\mu\nu}-2H\dot{\tilde{g}}_{\mu\nu}+2(\dot{H}-2H^2)\tilde{g}_{\mu\nu}=0;
\end{equation*}%
but we will not consider it, since we are seeking for selfconsistent choices, so we want that
the linearized Einsten Equations have a unique solution. From now on, we will call as \lq\lq the
linearized Einstein Equations\rq\rq the system (\ref{lin ein eq}).

Near $\tau_I$, the matter inhomogeneities cannot have yet generated the metric perturbations
$A, B, C$. For this reason, as initial conditions for (\ref{lin ein eq}) we ask that $A, B, C$ are zero at $\tau_I$.

\section{Newtonian gauge}

The previuous Theorem expresses the metric in the harmonic gauge, but what is the
suitable gauge for the comparison to the standard calculations? For the local effects
(about galaxies, clusters...) it is used the newtonian approximation, i.e. the newtonian
gauge. For the global effects, cosmologists assume a FRWL metric, which is diagonal.
Anyway, we have to compare our perturbed metric to a diagonal one, and the metric is
diagonalized in the newtonian gauge.

\subsection{Gauge transformation}

\begin{lem}
	Via the transformation $\tau'=\tau-B(\underline{x};\tau)$, the metric $g_{\mu\nu}$ is expressed in the
newtonian gauge as%
	\begin{equation}
	\label{newt g}
		g_{\mu'\nu'}=a(\tau')^2\left(\begin{matrix}
			2\Psi+1 & \vec{0} \\

			\vec{0} & (2\Phi-1)\delta_{ij}
		\end{matrix}\right),
	\end{equation}%
	where the gravitational potentials are%
	\begin{align}
		\Psi&=A+\dot{B}+HB \cr
		\Phi&=C-HB.
	\end{align}
\end{lem}

From now on we will use the newtonian coordinates, without writing the primes.
Notice that the second gauge condition (\ref{gauge}) guarantees that%
\begin{equation}
	\Psi\equiv\Phi.
\end{equation}

The fictitious effects of matter and dark energy are not independent from the gauge, and
this makes important the choice of the newtonian gauge.

\begin{oss}
\label{oss crosta}
The dependence on the gauge can be quite surprising, but it is coherent
with the Lusanna's line of research, e.g. \cite{Lusanna:2012kx}. The non diagonal component of the space-time
metric contributes to the relativistic effects of dark matter and dark energy; this was
recently confirmed for the dark matter halo of the Milky Way in \cite{Crosta:2018var}, where the rotation
of the galaxy generates a certain rotational $\vec{B}$.
\end{oss}

\subsection{Averaging the metric}

What we will compare with the CSM is just the average of the metric, since the metric
itself is not homogeneous and never allow for an exact equivalence. Such an average
depends only on time%
\begin{equation}
\label{aver g}
	\langle g_{\mu\nu}\rangle=a^2\left(\begin{matrix}
			2\langle\Psi\rangle+1 & \vec{0} \\

			\vec{0} & (2\langle\Phi\rangle-1)\delta_{ij}
		\end{matrix}\right),
\end{equation}%
where we know from the last Lemma%
\begin{align}
	\langle\Psi\rangle&=\langle A\rangle+\langle\dot{B}\rangle+H\langle B\rangle \cr
	\langle\Phi\rangle&=\langle C\rangle-H\langle B\rangle.
\end{align}

Now we recall another result from \cite{re}.

\begin{lem}
\label{lem2}
	Let us consider the Green functions for (\ref{lin ein eq})%
	\begin{align}
		\left(\Box-2H\partial_{\tau}+2(\dot{H}-2H^2)\right)G_{\tau'}(\underline{x};\tau)&=\delta^{(3)}(\underline{x})\delta(\tau-\tau') \cr
		\left(\Box-2H\partial_{\tau}\right)G^C_{\tau'}(\underline{x};\tau)&=\delta^{(3)}(\underline{x})\delta(\tau-\tau'),
	\end{align}%
	and let us assume the separation of variables for the matter inhomogeneity%
	\begin{equation}
		\tilde{\rho}(\underline{x};\tau):=\tilde{\rho}_0(\underline{x})T(\tau).
	\end{equation}

	Then we can express the average of metric distrortions as follows%
	\begin{equation}
		\langle A\rangle(\tau)=4\pi G\langle\tilde{\rho}_0\rangle u_A(\tau)=\frac{3}{2}\Omega_{IM0}H_0^2u_A(\tau),
	\end{equation}%
	s.t.%
	\begin{equation}
		u_A(\tau):=\int_{|\underline{r}|<R(\tau)}\int_{\tau_I}^{\tau}G_{\tau'}(\underline{r};\tau)a(\tau')^2T(\tau')d\tau'd^3\underline{r}.
	\end{equation}

	The separation of variables does not hold exactly for $A$, but let we can approximate%
	\begin{equation}
		A(\underline{x};\tau)\propto u_A(\tau).
	\end{equation}

	Then, in the same way%
	\begin{equation}
		\langle C\rangle(\tau)=\frac{3}{2}\Omega_{IM0}H_0^2(2u_{AC}(\tau)+u_C(\tau)),
	\end{equation}%
	s.t.%
	\begin{equation}
		u_{AC}(\tau):\cong\int_{|\underline{r}|<R(\tau)}\int_{\tau_I}^{\tau}G^C_{\tau'}(\underline{r};\tau)\dot{H}(\tau')u_A(\tau')d\tau'd^3\underline{r}
	\end{equation}%
	and%
	\begin{equation}
		u_C(\tau):=\int_{|\underline{r}|<R(\tau)}\int_{\tau_I}^{\tau}G^C_{\tau'}(\underline{r};\tau)a(\tau')^2T(\tau')d\tau'd^3\underline{r}.
	\end{equation}
\end{lem}

Here we use the $u$ functions to describe the time evolution of the perturbations.

The separation of variables for $\tilde{\rho}$ holds when there is a single component dominating. We can express it with the density contrast%
\begin{equation*}
	\tilde{\rho}_0T(\tau)=\tilde{\rho}:=\delta_M\bar{\rho}_M\propto\delta_Ma(\tau)^{-3}.
\end{equation*}

E.g. when the matter dominates, it is%
\begin{equation*}
	\delta_M\propto a \Rightarrow T(\tau)=a(\tau)^{-2},
\end{equation*}%
as we know from \cite{oliver}, \cite{peebles}, \cite{peebles2}.

When to dominate is dark energy, the matter structures are ripped apart with the same expansion rate of the universe%
\begin{equation*}
	\delta_M=cost. \Rightarrow T(\tau)=a(\tau)^{-3}.
\end{equation*}

When to dominate is radiation, the density contrast is well described by%
\begin{equation*}
	\delta_M\propto\ln\left(\frac{4}{y}\right), \; \; \; s.t. \; \; \; y:=\frac{a(\tau)}{a_{RM}},
\end{equation*}%
as \cite{oliver}, \cite{peebles}, \cite{peebles2} say again, and $a_{RM}$ is the value of $a$ for which the matter starts to dominate on the radiation; thus the $T$ function is%
\begin{equation*}
	T(\tau)=a(\tau)^{-3}\ln\left(\frac{4a_{RM}}{a(\tau)}\right)=a(\tau)^{-3}[\ln(4a_{RM})-\ln a(\tau)].
\end{equation*}

Since we are in the newtonian gauge, we need also the average of $B$. We can obtain it
averaging the second gauge condition (\ref{gauge}).

\begin{lem}
\label{lem3}
	\begin{equation}
		\langle B\rangle(\tau)=\frac{3}{2}\Omega_{IM0}H_0^2u_B(\tau),
	\end{equation}%
	s.t.%
	\begin{equation}
	\label{B}
		u_B(\tau):=a(\tau)^{-2}\int_{\tau_I}^{\tau}a(\tau')^2(2u_{AC}(\tau')+u_C(\tau')-u_A(\tau'))d\tau'.
	\end{equation}
\end{lem}
\begin{proof}
	We know that $\langle\dot{B}\rangle+2\frac{\dot{a}}{a}\langle B\rangle=\langle C\rangle-\langle A\rangle$. After expressing $\langle B\rangle(\tau):=a(\tau)^{-2}b(\tau)$, we obtain%
	\begin{align*}
	&a(\tau)^{-2}\dot{b}(\tau)=\langle C\rangle-\langle A\rangle=\frac{3}{2}\Omega_{IM0}H_0^2(2u_{AC}(\tau)+u_C(\tau)-u_A(\tau)) \Rightarrow \cr
	&b(\tau)=\frac{3}{2}\Omega_{IM0}H_0^2\int_{\tau_I}^{\tau}a(\tau')^2(2u_{AC}(\tau')+u_C(\tau')-u_A(\tau'))d\tau',
	\end{align*}%
	which proves the assertion.
\end{proof}

\subsection{Formulas for the fictitious components}

The fictitious components are determined by (\ref{q0 sys}). We can rewrite it using the auxiliary
variables \emph{ract} and \emph{sum}, defined as in \cite{re}.

\begin{equation}
\label{ract sum}
	\begin{cases}
		sum\cdot\Omega_{IM0}+o(\Omega_{IM0}):=\Omega_{FM0}+\Omega_{F\Lambda0}=1-\sum_w\Omega_{Tw0} \\

		ract\cdot\Omega_{IM0}+o(\Omega_{IM0}):=\Omega_{FM0}-2\Omega_{F\Lambda0}=2q_0-\sum_w(1+3w)\Omega_{Tw0}
	\end{cases}.
\end{equation}

For an evaluation of these, we need to know the perturbations of on $q_0$ and $\Omega_{Tw0}$. The
magnitude of the perturbations is determined by the comparison with the CSM metric%
\begin{align}
\label{tildi}
	&a^2[(2\langle\Psi\rangle+1)d\tau^2+(2\langle\Phi\rangle-1)\delta_{ij}dx_idx_j]=\langle g_{\mu\nu}\rangle:=dt^2-\bold{a}^2\delta_{ij}dx_idx_j \Rightarrow \cr
	&dt=\tilde{t}\cdot a d\tau:=\sqrt{1+2\langle\Psi\rangle}\cdot a d\tau, \; \; \; \bold{a}=\tilde{a}\cdot a:=\sqrt{1-2\langle\Phi\rangle}\cdot a.
\end{align}

The conditions at the present time are%
\begin{equation}
\label{alg sys}
	\begin{cases}
		\bold{a}(t_0):=1 \\

		\bold{H}_0=\partial_t\bold{a}|_{t_0} \\

		q_0=-\frac{\partial_t^2\bold{a}}{\bold{H}}|_{t_0}
	\end{cases}.
\end{equation}%
From the first of these, we obtain the value of $a_0:=a(t_0)\neq a(\tau_0)=1$, since%
\begin{align*}
	\frac{1}{a_0}&=\tilde{a}_0=\sqrt{1-2\langle C\rangle_0-2H_0\langle B\rangle_0} \Rightarrow \cr
	a_0&=1+\langle C\rangle_0+H_0\langle B\rangle_0+o(\Omega_{IM0}).
\end{align*}

Now, we can consider $a$ as the time variable. By now, we denote with a prime the
derivatives with respect to $a$. From (\ref{tildi})%
\begin{equation*}
	dt=\tilde{t}a d\tau=\tilde{t}\frac{a}{\partial_{\tau}a}\partial_{\tau}a=\frac{\tilde{t}}{H}da
\end{equation*}%
so that for any given quantity $Q$ depending on the time, we have%
\begin{equation*}
	Q':=\frac{dQ}{da}=\frac{H}{\tilde{t}}\frac{dQ}{d\tau}=\frac{H}{\tilde{t}}\dot{Q}.
\end{equation*}

Using the relation in (\ref{alg sys}), we can find firstly the perturbations of $\Omega_{Tw0}$. Indeed, form definition (\ref{bg comp})%
\begin{equation}
\label{Tw0}
	\Omega_{Hw}=\frac{8\pi G}{3\bold{H}_0^2}\bold{a}^2\bar{\rho}_{w0}=\left(\frac{H_0}{\bold{H}_0}\right)^2\tilde{a}^2\bar{\Omega}_w.
\end{equation}%
From the second equation in (\ref{alg sys}), we can compute%
\begin{equation}
	\left(\frac{H_0}{\bold{H}_0}\right)^2\tilde{a}_0^2=1+2[\langle A\rangle_0+H_0\langle B\rangle_0-H'_0\langle B\rangle_0+\langle C\rangle'_0]+o(\Omega_{IM0}).
\end{equation}%
For any $w$ it's $\Omega_{Tw0}=\Omega_{Hw0}$, with the exception of $\Omega_{TM0}=\Omega_{HM0}+\Omega_{IM0}$. Thus
\begin{align*}
	sum\cdot\Omega_{IM0}+o(\Omega_{IM0})&=1-\sum_w\Omega_{Tw0}=1-\Omega_{IM0}-\left(\frac{H_0}{\bold{H}_0}\right)^2\tilde{a}_0^2\sum_w\bar{\Omega}_w \Rightarrow \cr
	(sum+1)\Omega_{IM0}+o(\Omega_{IM0})&=1-[1+2(\langle A\rangle_0+H_0\langle B\rangle_0-H'_0\langle B\rangle_0+\langle C\rangle'_0)+o(\Omega_{IM0})]\cdot1=\cr
	&=-2[\langle A\rangle_0+H_0\langle B\rangle_0-H'_0\langle B\rangle_0+\langle C\rangle'_0]+o(\Omega_{IM0}).
\end{align*}

As for the perturbation of $q_0$, we must remember that its zeroth order part is not zero, in
general, but the background has a deceleration%
\begin{equation}
	\bar{q}_0=\frac{1}{2}\sum_w(1+3w)\bar{\Omega}_{w0}.
\end{equation}%
It is distorted by the perturbation, then at first order we expect to have%
\begin{equation}
	q_0:=\bar{q}_0+q_{\Omega}\Omega_{IM0}+o(\Omega_{IM0})
\end{equation}%
for some coefficient $q_{\Omega}$. We can compute each of these from the third of (\ref{alg sys}), obtaining%
\begin{equation}
	q_0=-\frac{H'_0}{H_0}+\left[\langle A\rangle'_0+2\langle C\rangle'_0+\langle C\rangle''_0-\frac{H'_0}{H_0}(\langle C\rangle_0+\langle C\rangle'_0-H'_0\langle B\rangle_0)-H''_0\langle B\rangle_0\right]+o(\Omega_{IM0}).
\end{equation}%
In particular, this means that%
\begin{equation}
\label{q00}
	\frac{1}{2}\sum_w(1+3w)\bar{\Omega}_{w0}=\bar{q}_0=-\frac{H'_0}{H_0}.
\end{equation}

Together with (\ref{Tw0}), this gives%
\begin{align*}
	\frac{1}{2}ract\cdot\Omega_{IM0}+o(\Omega_{IM0})&=q_0-\frac{1}{2}\sum_w(1+3w)\Omega_{Tw0} \\
	&=[\bar{q}_0+q_{\Omega}\Omega_{IM0}+o(\Omega_{IM0})]-\frac{1}{2}(1+3w)|_{w=0}\Omega_{IM0} \\
	&-[1-(sum+1)\Omega_{IM0}+o(\Omega_{IM0})]\frac{1}{2}\sum_w(1+3w)\bar{\Omega}_{w0} \Rightarrow \\
	\frac{1}{2}(ract+1)\Omega_{IM0}+o(\Omega_{IM0})&=(sum+1)\Omega_{IM0}\bar{q}_0+q_{\Omega}\Omega_{IM0}+o(\Omega_{IM0}).
\end{align*}

\begin{thm}
	At first order, the effects of the matter inhomogeneities can be interpreted
in terms of total fictitious components%
	\begin{equation}
		\begin{cases}
			\Omega_{FM0}=\frac{2sum+ract}{3}\Omega_{IM0}+o(\Omega_{IM0}) \\

			\Omega_{F\Lambda0}=\frac{sum-ract}{3}\Omega_{IM0}+o(\Omega_{IM0})
		\end{cases},
	\end{equation}%
	where the auxiliary quantities are%
	\begin{equation}
	\label{sum ract}
		\begin{cases}
			\frac{1}{2}(sum+1)\Omega_{IM0}=&-\langle A\rangle_0-H_0\langle B\rangle_0+H'_0\langle B\rangle_0-\langle C\rangle'_0 \\

			\frac{1}{2}(ract+1)\Omega_{IM0}=&\langle A\rangle'_0+2\langle C\rangle'_0+\langle C\rangle''_0-H''_0\langle B\rangle_0 \\
			&+\frac{H'_0}{H_0}(2\langle A\rangle_0+2H_0\langle B\rangle_0-H'_0\langle B\rangle_0+\langle C\rangle'_0-\langle C\rangle_0)
		\end{cases}.
	\end{equation}
\end{thm}

\section{ODEs for the metric perturbations}

Now we should solve (\ref{lin ein eq}), replacing the resultant $A, B, C$ inside (\ref{sum ract}). The general
PDEs (\ref{lin ein eq}) are a formidable mathematical task. We obtained in \cite{re} an exact solution for
the case with constant coefficients, but it seems to be impossible an analytical solution
when the coefficients depend on $\tau$. However, here we are investigating only the global
effects, which depend only on the average of $A, B, C$, as (\ref{sum ract}) shows. Performing a spatial
averaging procedure on the PDEs (\ref{lin ein eq}), these are replaced by simpler ODEs, depending on the
time only, for $\langle A\rangle, \langle B\rangle, \langle C\rangle$. Such ODEs admit analytical solutions.

\subsection{Reduction of the dimensions}

From Lemmas \ref{lem2} and \ref{lem3} we have%
\begin{equation*}
	\langle A\rangle=\frac{3}{2}\Omega_{IM0}H_0^2u_A(\tau), \; \; \; \langle B\rangle=\frac{3}{2}\Omega_{IM0}H_0^2u_B(\tau), \; \; \; \langle C\rangle\cong\frac{3}{2}\Omega_{IM0}H_0^2(2u_{AC}(\tau)+u_C(\tau));
\end{equation*}%
where%
\begin{align*}
	u_A(\tau)&=\int_{\tau_I}^{\tau}\left[\int_{|\underline{r}<R(\tau)|}G_{\tau'}(\underline{r};\tau)d^3\underline{r}\right]a(\tau')^2T(\tau')d\tau', \\
	u_{AC}(\tau)&=\int_{\tau_I}^{\tau}\left[\int_{|\underline{r}<R(\tau)|}G^C_{\tau'}(\underline{r};\tau)d^3\underline{r}\right]\dot{H}(\tau')u_A(\tau')d\tau', \\
	u_C(\tau)&=\int_{\tau_I}^{\tau}\left[\int_{|\underline{r}<R(\tau)|}G^C_{\tau'}(\underline{r};\tau)d^3\underline{r}\right]a(\tau')^2T(\tau')d\tau', \\
	u_B(\tau)&=a(\tau)^{-2}\int_{\tau_I}^{\tau}a(\tau')^2\left(2u_{AC}(\tau')+u_C(\tau')-u_A(\tau')\right)d\tau';
\end{align*}%
and%
\begin{align*}
	\left(\Box-2H\partial_{\tau}+2(\dot{H}-2H^2)\right)G_{\tau'}(\underline{x};\tau)&=\delta^{(3)}(\underline{x})\delta(\tau-\tau'), \\
	\left(\Box-2H\partial_{\tau}\right)G^C_{\tau'}(\underline{x};\tau)&=\delta^{(3)}(\underline{x})\delta(\tau-\tau').
\end{align*}

Since the Green functions are symmetric under spatial rotation, we can reduce the spatial
dimensions to one%
\begin{align}
	&G_{\tau'}(\underline{r};\tau)=-\frac{1}{2\pi|\underline{r}|}\partial_r\Gamma_{\tau'}(|\underline{r}|;\tau) \quad s.t. \cr
	&\left(\partial_r^2-\partial_{\tau}^2-2H\partial_{\tau}+2(\dot{H}-2H^2)\right)\Gamma_{\tau'}(r;\tau)=\delta(r)\delta(\tau-\tau');
\end{align}%
and the same for $G^C_{\tau'}(\underline{r};\tau)$.%
This allow us to express in another way the terms as%
\begin{align*}
	\int_{|\underline{r}<R(\tau)|}G_{\tau'}(\underline{r};\tau)d^3\underline{r}&=\int_0^{R(\tau)}\left[-\frac{1}{2\pi r}\partial_r\Gamma_{\tau'}(r;\tau)\right]4\pi r^2dr \\
	&=-2\int_0^{R(\tau)}r\partial_r\Gamma_{\tau'}(r;\tau)dr \cr
	&=2\left(\int_0^{R(\tau)}\Gamma_{\tau'}(r;\tau)dr-[r\Gamma_{\tau'}(r;\tau)]_{r=0}^{R(\tau)}\right).
\end{align*}

Now let us define the auxiliary field%
\begin{equation}
	v_A(r;\tau):=\int_{\tau_I}^{\tau}\Gamma_{\tau'}(r;\tau)a(\tau')^2T(\tau')d\tau';
\end{equation}%
and similar for $v_{AC}$ and $v_C$. Then, we can prove

\begin{lem}
	The metric perturbations evolve as%
	\begin{align}
		u_A(\tau)&=\int_{-R(\tau)}^{R(\tau)}v_A(r;\tau)dr-2[rv_A(r;\tau)]_{r=0}^{R(\tau)}, \cr
		u_{AC}(\tau)&=\int_{-R(\tau)}^{R(\tau)}v_{AC}(r;\tau)dr-2[rv_{AC}(r;\tau)]_{r=0}^{R(\tau)}, \cr
		u_C(\tau)&=\int_{-R(\tau)}^{R(\tau)}v_C(r;\tau)dr-2[rv_C(r;\tau)]_{r=0}^{R(\tau)};
	\end{align}%
	where the $v$ fields solve the 2D PDEs%
	\begin{align}
	\label{2D PDEs}
		\left(\partial_r^2-\partial_{\tau}^2-2H\partial_{\tau}+2(\dot{H}-2H^2)\right)v_A(r;\tau)&=\delta(r)a(\tau)^2T(\tau), \cr
		\left(\partial_r^2-\partial_{\tau}^2-2H\partial_{\tau}\right)v_{AC}(r;\tau)&=\delta(r)\dot{H}(\tau)u_A(\tau), \cr
		\left(\partial_r^2-\partial_{\tau}^2-2H\partial_{\tau}\right)v_C(r;\tau)&=\delta(r)a(\tau)^2T(\tau).
	\end{align}
\end{lem}
\begin{proof}
	Let us start from the time derivatives of $v_A$.

	\begin{align*}
		\dot{v}_A(r;\tau)&=[\Gamma_{\tau'}(r;\tau)a(\tau')^2T(\tau')]|_{\tau'=\tau}+\int_{\tau'}^{\tau}\partial_{\tau}\Gamma_{\tau'}(r;\tau)a(\tau')^2T(\tau')d\tau', \cr
		\ddot{v}_A(r;\tau)&=\partial_{\tau}[\Gamma_{\tau'}(r;\tau)a(\tau')^2T(\tau')]|_{\tau'=\tau}+[\partial_{\tau}\Gamma_{\tau'}(r;\tau)a(\tau')^2T(\tau')]|_{\tau'=\tau} \cr
		&+\int_{\tau'}^{\tau}\partial_{\tau}^2\Gamma_{\tau'}(r;\tau)a(\tau')^2T(\tau')d\tau'.
	\end{align*}%
	$\Gamma_{\tau'}(r;\tau)$ satisfies a wave equation, so it holds a causality principle%
	\begin{equation*}
		\forall |r|>\tau-\tau': \Gamma_{\tau'}(r;\tau)\equiv0.
	\end{equation*}%
	Setting $\tau'=\tau$ it becomes%
	\begin{equation*}
		\forall |r|>0: \Gamma_{\tau}(r;\tau)\equiv0.
	\end{equation*}%
	Continuity in $r=0$ requires%
	\begin{equation*}
		[\Gamma_{\tau'}(r;\tau)a(\tau')^2T(\tau')]|_{\tau'=\tau}=\Gamma_{\tau}(r;\tau)a(\tau)^2T(\tau)\equiv0;
	\end{equation*}%
	so that the boundary terms of $\dot{v}_A$ and $\ddot{v}_A$ must vanish. Now, we can check by substitution%
	\begin{align*}
		\partial_r^2v_A-\ddot{v}_A&-2H\dot{v}_A+2(\dot{H}-2H^2)v_A \cr
		&=\int_{\tau_I}^{\tau}[\partial_r^2\Gamma_{\tau'}-\partial_{\tau}^2\Gamma_{\tau'}-2H\partial_{\tau}\Gamma_{\tau'}+2(\dot{H}-2H^2)\Gamma_{\tau'}]a(\tau')^2T(\tau')d\tau' \cr
		&=\int_{\tau_I}^{\tau}\delta(r)\delta(\tau-\tau')a(\tau')^2T(\tau')d\tau'=\delta(r)a(\tau)^2T(\tau).
	\end{align*}

	Notice that this PDE is symmetric under $r\rightarrow-r$, so that $v_A(r;\tau)=v_A(-r;\tau)$. This allows
us to write%
	\begin{align*}
		u_A(\tau)&=\int_{\tau_I}^{\tau}\left[\int_{|\underline{r}<R(\tau)|}G_{\tau'}(\underline{r};\tau)d^3\underline{r}\right]a(\tau')^2T(\tau')d\tau' \cr
		&=2\int_{\tau_I}^{\tau}\left(\int_0^{R(\tau)}\Gamma_{\tau'}(r;\tau)dr-[r\Gamma_{\tau'}(r;\tau)]_{r=0}^{R(\tau)}\right)a(\tau')^2T(\tau')d\tau' \cr
		&=2\left(\int_0^{R(\tau)}v_A(r;\tau)dr-[rv_A(r;\tau)]_{r=0}^{R(\tau)}\right) \cr
		&=\int_{-R(\tau)}^{R(\tau)}v_A(r;\tau)dr-2[rv_A(r;\tau)]_{r=0}^{R(\tau)}.
	\end{align*}%
	The proof is analogous for $v_{AC}$ and $v_C$.
\end{proof}

\subsection{Fourier transform}

We can eliminate the derivatives w.r.t. the spatial variable $r$, by writing (\ref{2D PDEs}) for the $v$s
in Fourier transform. If we define%
\begin{equation}
	v_A(r;\tau):=\frac{1}{2\pi}\int\hat{v}_A(\omega;\tau)e^{-ir\omega}d\omega;
\end{equation}%
then, the corresponding PDE becomes%
\begin{equation}
\label{hat PDE}
	\left(-\omega^2-\partial_{\tau}^2-2H\partial_{\tau}+2(\dot{H}-2H^2)\right)\hat{v}_A(\omega;\tau)=a(\tau)^2T(\tau).
\end{equation}%
The analogous holds for the other $v$s. Now we manipulate the term in the $u$s.

\begin{lem}
\label{lem int}
	\begin{align}
		\int_{-R(\tau)}^{R(\tau)}v_A(r;\tau)dr&=\hat{v}_A(\omega;\tau)|_{\omega=0}, \cr
		\int_{-R(\tau)}^{R(\tau)}v_{AC}(r;\tau)dr&=\hat{v}_{AC}(\omega;\tau)|_{\omega=0}, \cr
		\int_{-R(\tau)}^{R(\tau)}v_C(r;\tau)dr&=\hat{v}_C(\omega;\tau)|_{\omega=0}.
	\end{align}
\end{lem}
\begin{proof}
	First, consider that $v_A$ satisfies the wave equation (\ref{2D PDEs}), so that we can impose the causality condition%
	\begin{equation*}
		\forall |r|>\tau-\tau_I=R(\tau): v_A(r;\tau)\equiv0.
	\end{equation*}%
	Therefore%
	\begin{equation*}
		\int_{-R(\tau)}^{R(\tau)}v_A(r;\tau)dr=\int_{-\infty}^{+\infty}v_A(r;\tau)dr.
	\end{equation*}%
	After applying the Fourier transform and switching the integrals we find%
	\begin{align*}
		\int v_A(r;\tau)dr&=\int\left[\frac{1}{2\pi}\int\hat{v}_A(\omega;\tau)e^{-ir\omega}d\omega\right]dr\\
		&=\int\left[\int\frac{1}{2\pi}e^{-ir\omega}dr\right]\hat{v}_A(\omega;\tau)d\omega=\int\delta(\omega)\hat{v}_A(\omega;\tau)d\omega;
	\end{align*}%
	which proves the assertion. For the others $v$s the proof is analogous.
\end{proof}

Now, we need to evaluate the boundary terms $[rv(r;\tau)]_{r=0}^{R(\tau)}$.

\begin{lem}
\label{lem 0}
	The term $rv(r;\tau)|_{r=0}$ always vanishes.
\end{lem}
\begin{proof}
	For Fourier properties%
	\begin{equation*}
		rv_A(r;\tau)|_{r=0}=-i\int\partial_{\omega}\hat{v}_A(\omega;\tau)e^{-ir\omega}d\omega|_{r=0}=-i\int\partial_{\omega}\hat{v}_A(\omega;\tau)d\omega=-i[\hat{v}_A(\omega;\tau)]_{-\infty}^{+\infty}.
	\end{equation*}%
	For an evaluation of $\hat{v}_A(\omega;\tau)$ when $\omega$ goes to infinity, we can manipulate the corresponding ODE%
	\begin{align*}
		\hat{v}_A(\omega;\tau)&=\frac{1}{\omega^2}[(-\partial_{\tau}^2-2H\partial_{\tau}+2(\dot{H}-2H^2))\hat{v}_A(\omega;\tau)-a(\tau)^2T(\tau)] \\
		&\sim^{\omega\rightarrow\pm\infty}\frac{1}{\omega^2}(-\partial_{\tau}^2-2H\partial_{\tau}+2(\dot{H}-2H^2))\hat{v}_A(\omega;\tau).
	\end{align*}%
	A solution is%
	\begin{equation*}
		\hat{v}_A(\omega;\tau)\sim^{\omega\rightarrow\pm\infty}0;
	\end{equation*}%
	which proves the assertion. The proof is analogous for the others $v$s.
\end{proof}

For the other term, we don't need the Fourier transform.

\begin{lem}
\label{lem R}
	The term $rv(r;\tau)|_{r=R(\tau)}$ vanishes if and only if $\tau_I>-\infty$ and $a(\tau)^2T(\tau)\in L^1_{loc}([\tau_I;\tau_F))$. Otherwise, it it divergent.
\end{lem}
\begin{proof}
	We know that $v_A$ satisfies a wave equation (\ref{2D PDEs}), whose principal symbol is the same
as for a 2D d'alembertian. As in \cite{re} \S4, near the wave boundary $r\rightarrow R(\tau)$ the solution
depends on the principal symbol only, and we can neglect the terms $-2H\dot{v}_A+2(\dot{H}-2H^2)$
in that asymptotic region:%
	\begin{equation*}
		v_A(r;\tau)\sim^{r\rightarrow R(\tau)}\bar{v}_A(r;\tau) \; \; \; s.t. \; \; \; (\partial_r^2-\partial_{\tau}^2)\bar{v}_A(r;\tau)=\delta(r)a(\tau)^2T(\tau).
	\end{equation*}%
	It is easily solved by%
	\begin{equation*}
		\bar{v}_A(r;\tau)=S(\tau-|r|) \; \; \; s.t. \; \; \; S(x):=\int_{\tau_I}^xa(\tau)^2T(\tau)d\tau.
	\end{equation*}%
	Notice that $S$ diverges if $a(\tau)^2T(\tau)\not\in L^1_{loc}([\tau_I;\tau_F))$. After replacing it%
	\begin{align*}
		rv_A(r;\tau)|_{r=R(\tau)}&=\lim_{r\rightarrow R(\tau)}r\bar{v}_A(r;\tau)=\lim_{r\rightarrow\tau-\tau_I}rS(\tau-|r|) \\
		&=\lim_{x\rightarrow\tau_I}(\tau-x)S(x)=\lim_{x\rightarrow\tau_I}(\tau-x)\int_{\tau_I}^xa(\tau)^2T(\tau)d\tau.
	\end{align*}

	Let now consider the case $\tau_I>-\infty$, so that the requirement on $S$ becomes $a(\tau)^2T(\tau)\in L^1_{loc}([0;\tau_F))$.
Then, the integral $\int_{\tau_I}^xa(\tau)^2T(\tau)d\tau$ goes to zero and $rv_A(r;\tau)|_{r=R(\tau)}\equiv0$.%
	On the other hand, in the case $\tau_I=-\infty$ we see that $\tau-x\rightarrow+\infty$. Since $a(\tau)^2T(\tau)$ is always
positive, we get the divergence $rv_A(r;\tau)|_{r=R(\tau)}\equiv+\infty$.

	The proof is analogous for $v_C$. For $v_{AC}$ we obtain%
	\begin{equation*}
		rv_{AC}(r;\tau)|_{r=R(\tau)}=\lim_{x\rightarrow\tau_I}(\tau-x)\int_{\tau_I}^x\dot{H}(\tau)u_A(\tau)d\tau.
	\end{equation*}%
	As before, it diverges if $\tau_I=-\infty$. If $\tau_I>-\infty$ but $a(\tau)^2T(\tau)\not\in L^1_{loc}([\tau_I;\tau_F))$, then
$u_A(\tau)=\int_{-R(\tau)}^{R(\tau)}v_A(r;\tau)dr-2rv_A(r;\tau)|_{r=R(\tau)}\equiv-\infty$ as we saw, and also $v_{AC}$ diverges. If
we have $a(\tau)^2T(\tau)\in L^1_{loc}([\tau_I;\tau_F))$, then $u_A(\tau)=\hat{v}_A(0;\tau)$ for Lemmas \ref{lem int} and \ref{lem 0}; it
converges, and the first of (\ref{lin ein eq}) assures that $\dot{H}(\tau)u_A(\tau)\in L^1_{loc}([\tau_I;\tau_F))$.
\end{proof}

Putting the last Lemmas all together, we obtain the $u$s as solutions of some ODEs.

\begin{thm}
	 If $\tau_I>-\infty$ and $a(\tau)^2T(\tau)\in L^1_{loc}([\tau_I;\tau_F))$, then%
	\begin{align}
	\label{ODEs}
			\ddot{u}_A(\tau)+2H\dot{u}_A(\tau)+2(2H^2-\dot{H})u_A(\tau)&=-a(\tau)^2T(\tau) \cr
			\ddot{u}_{AC}(\tau)+2H\dot{u}_{AC}(\tau)&=-\dot{H}(\tau)u_A(\tau) \cr
			\ddot{u}_C(\tau)+2H\dot{u}_C(\tau)&=-a(\tau)^2T(\tau).
	\end{align}%
	Otherwise, $\langle A\rangle$, $\langle B\rangle$, $\langle C\rangle$ always diverge.
\end{thm}
\begin{proof}
	From Lemmas \ref{lem int}, \ref{lem 0} and \ref{lem R}, all the $u$s are%
	\begin{equation*}
		u(\tau)=\int_{-R(\tau)}^{R(\tau)}v(r;\tau)dr=\hat{v}(\omega;\tau)|_{\omega=0}.
	\end{equation*}%
	The $\hat{v}$s obey equations like (\ref{hat PDE}). Writing them for the $u$s we have the assertion.
\end{proof}

\section{Single component cases}

It is still impossible to solve analytically the evolution (\ref{bg evol}) for $a$ and the ODEs (\ref{ODEs}) for
a general choice of components $\{\bar{\Omega}_{w0}\}_w$. Moreover, for such a general choice it's quite
difficult to determine the form of the source $\tilde{\rho}\propto T(\tau)$. However, we are able to solve
exactly the equations when a single component $\bar{\Omega}_{w}$ dominates. We can approximate the
general evolution as a succession of \lq\lq epochs\rq\rq; during each epoch, we consider just the dominant component%
\begin{equation*}
	\forall\tau | \bar{\Omega}_{w}(\tau)=\max_{w'}\bar{\Omega}_{w'}(\tau) : \bar{\Omega}_{w0}\cong1,
\end{equation*}%
so that each epoch has a single-component evolution. The full evolution is obtained sticking
the partial functions, imposing that $a(\tau)\in C^0(\tau_I;\tau_F)$, since (\ref{bg evol}) is first order,
$\langle A\rangle, \langle C\rangle\in C^1(\tau_I;\tau_F)$, since (\ref{ODEs}) are second order; and $\langle B\rangle\in C^0(\tau_I;\tau_F)$,
because $u_B$ is obtained by an integral in (\ref{B}).

\subsection{The First Selfconsistence Condition}

Let's start solving (\ref{bg evol}) for a general epoch with $\bar{\Omega}_{w'0}\cong\delta_{w',w}$.%
\begin{equation*}
	\left(\frac{\dot{a}}{H_0}\right)^2=a^{1-3w} \Rightarrow
\end{equation*}
\begin{equation}
\label{single a}
	a(\tau)=\begin{cases}
		\left(\frac{1}{\alpha}H_0(\tau-c)\right)^{\alpha}\quad w\neq-\frac{1}{3} \\

		e^{H_0(\tau-c)} \quad w=-\frac{1}{3}
	\end{cases} \quad s.t. \quad \alpha(w):=\frac{2}{1+3w};
\end{equation}%
where $c$ is an integration constant. We get immediately the coefficients of (\ref{ODEs})

\begin{align}
	H(\tau)&=\begin{cases}
		\frac{\alpha}{\tau-c} \quad w&\neq-\frac{1}{3} \\
		H_0 \quad w&=-\frac{1}{3}
	\end{cases}, \\
	2H&=\begin{cases}
		2\frac{\alpha}{\tau-c} \quad w&\neq-\frac{1}{3} \\
		2H_0 \quad w&=-\frac{1}{3}
	\end{cases}, \\
	2(2H^2-\dot{H})&=\begin{cases}
		2\frac{2\alpha^2+\alpha}{(\tau-c)^2} \quad w&\neq-\frac{1}{3} \\
		4H_0^2 \quad w&=-\frac{1}{3}
	\end{cases}.
\end{align}

Recalling (\ref{bg comp}) and that $a(\tau)$ is increasing (at least) near $\tau_I$, we see that the epochs
must be in order of decreasing $w$. In particular, during the first epoch it dominates
$w_M:=\max\{w\}$. Setting the initial condition%
\begin{equation*}
	\lim_{\tau\rightarrow\tau_I}a(\tau)=0 \Rightarrow
\end{equation*}
\begin{equation}
	\tau_I=\begin{cases}
		-\infty \quad \alpha(w_M)&<0 \vee w=-\frac{1}{3} \\

		c \quad \alpha(w_M)&>0
	\end{cases}.
\end{equation}%
By definition it is always $\alpha\neq0$ for definition. From the previous Theorem, we get immediately

\begin{cor} [First Selfconsistence Condition]
	A selfconsistent choice of components must be such that $w_M>-\frac{1}{3}$.
\end{cor}

In particular, a selfconsistent universe must develop the metric perturbations as described
by (\ref{ODEs}), with non constant coefficients.
\begin{oss}
	In \cite{re} we studied the costant coefficient case, filling the universe with an
exotic component s.t. $w=-\frac{1}{3}$. This breaks the First Selfconsistence Condition, which
explains the divergences we found in \cite{re} \S4.3: it is the contribution of $rv(r;\tau)|_{r=R(\tau)}\equiv\infty$.
It is possible to extract finite results even when the I SC is broken, as we did with a
renormalization via analytic continuation. A general renormalization method could be to
always neglect the term $rv(r;\tau)|_{r=R(\tau)}\equiv\infty$, using (\ref{ODEs}) for any $w_M$.
\end{oss}

As long as the I SC holds, we can fix $\tau_I:=0$ without lost of generality.

\subsection{Decoupling}

As we say in Lemma \ref{lem2}, for general coefficients of (\ref{lin ein eq}) we have just an approximated solution
of $\langle C\rangle$. This is due to the coupling between $C$ and $A$. Another advantage of the single component
evolution is to allow the decoupling the PDEs of $A$ and $C$%
\begin{equation*}
	\begin{cases}
		\Box A-2\alpha\frac{\dot{A}}{\tau-c}-2\alpha(2\alpha+1)\frac{A}{(\tau-c)^2}=4\pi Ga^2\tilde{\rho} \\

		\Box C-2\alpha\frac{\dot{C}}{\tau-c}+2\alpha\frac{A}{(\tau-c)^2}=4\pi Ga^2\tilde{\rho}
	\end{cases}.
\end{equation*}%
Let $\alpha\neq-\frac{1}{2}$.\footnote{The case $\alpha\neq-\frac{1}{2}$ happens only for the exotic component $w=-\frac{5}{3}$.} Then it is convenient to
define the auxiliary field%
\begin{equation}
	D:=A+(2\alpha+1)C,
\end{equation}%
which must satisfy the PDE%
\begin{equation}
	\Box D-2\frac{\alpha}{\tau-c}\dot{D}=8(\alpha+1)\pi Ga^2\tilde{\rho}.
\end{equation}%
All the results in \S5 hold true for $D$, so that%
\begin{align}
	&\langle D\rangle=3\Omega_{IM0}H_0^2u_D(\tau) \quad s.t. \cr
	&\ddot{u}_D+2\frac{\alpha}{\tau-c}\dot{u}_D=-(\alpha+1)a(\tau)^2T(\tau).
\end{align}%
From these we get an exact formula for $\langle C\rangle$%
\begin{equation}
	\langle C\rangle=\frac{\langle D\rangle-\langle A\rangle}{2\alpha+1}=\frac{3}{2}\Omega_{IM0}H_0^2\frac{2u_D(\tau)-u_A(\tau)}{2\alpha+1}.
\end{equation}

\begin{oss}
	Notice that in the dark energy epoch $\alpha=-1$ and the ODE for $u_D$ is free
of source. However, this doesn't imply that $u_D$ is zero, thus in general $\langle C\rangle\neq\langle A\rangle$.
\end{oss}

\subsection{Solving the ODEs}

To solve (\ref{ODEs}) for a general $w$, we need the form of $T(\tau)$. We will assume%
\begin{equation}
	\delta_M\propto a(\tau)^n,
\end{equation}%
with $n(w)$ a regular function, of which we know $n(0)=1$ and $n(-1)=0$. This
assumption does not certainly hold for the radiation epoch ($w=\frac{1}{3}$), when%
\begin{equation}
	\delta_M\propto\ln(4a_R)-\ln a(\tau) \; \; \; s.t. \; \; \; a_R=\max\{a(\tau)|\bar{\Omega}_{R}(\tau)=\max\bar{\Omega}_w(\tau)\}.
\end{equation}%
Let us start by solving for $u_A$. In general, it has a term $u_{IA}$ generated by the source $-a(\tau)^2T(\tau)=-a(\tau)^{n-1}$,
and a term $u_{HA}$ without sources. They result to be%
\begin{align}
	u_{IA}(\tau)=&H_0^{-2}u_{A0}(H_0\tau)^{n_A} \quad s.t. \quad n_A=(n-1)\alpha+2 \cr
	&and \quad u_{A0}=-\frac{\alpha^{(1-n)\alpha}}{(n\alpha-\alpha+2)(n\alpha+\alpha+1)+2\alpha(2\alpha+1)}, \\
	u_{HA}(\tau)\propto&(H_0\tau)^{n_H} \quad s.t. \quad n_H^2+(2\alpha-1)n_H+(4\alpha^2+2\alpha)=0.
\end{align}%
The exponent of $u_{HA}$ is%
\begin{equation}
	n_H=\left(\frac{1}{2}-\alpha\right)\pm\sqrt{\frac{1}{4}-3\alpha-3\alpha^2}.
\end{equation}%
It has an imaginary part if and only if%
\begin{align*}
	\alpha&\in\left(-\infty; \alpha(w_+):=-\frac{1}{\sqrt{3}}-\frac{1}{2}\right)\sqcup\left(\alpha(w_-):=\frac{1}{\sqrt{3}}-\frac{1}{2};+\infty\right) \Leftrightarrow \cr
	\Leftrightarrow w&\in(w_-\cong-0.9521;w_+\cong8.2855).
\end{align*}

Because of the arbitrariness of the integration constants $c_1$ and $c_2$, we can write in general%
\begin{align}
\label{uA}
	H_0^2&u_A(\tau)=u_{A0}(H_0\tau)^{(n-1)\alpha+2}\cr
	&+\begin{cases}
		\left[c_{A1}\sin(\sqrt{\xi}\cdot\ln H_0\tau)+c_{A2}\cos(\sqrt{\xi}\cdot\ln H_0\tau)\right](H_0\tau)^{\frac{1}{2}-\alpha} & w\in(w_;w_+) \\
		\left[c_{A1}(H_0\tau)^{\sqrt{-\xi}}+c_{A2}(H_0\tau)^{-\sqrt{-\xi}}\right](H_0\tau)^{\frac{1}{2}-\alpha} & w\not\in(w_;w_+)
	\end{cases},
\end{align}%
where $\xi:=3\alpha^2+3\alpha-\frac{1}{4}$. The solution for $D$ is simpler.%
\begin{align}
	&H_0u_D(\tau)=u_{D0}(H_0\tau)^{(n-1)\alpha+2}+c_{D1}(H_0\tau)^{1-2\alpha}+c_{D2} \cr
	&s.t. \quad u_{D0}=-\frac{\alpha^{(1-n)\alpha}}{(n\alpha-\alpha+2)(n\alpha+\alpha+1)}.
\end{align}%
Using this, we get for $C$%
\begin{align}
\label{uC}
	H_0^2&u_C(\tau)=u_{C0}(H_0\tau)^{(n-1)\alpha+2}+c_{D1}(H_0\tau)^{1-2\alpha}+c_{D2} \cr
	&+\begin{cases}
		-\frac{1}{2\alpha+1}\left[c_{A1}\sin(\sqrt{\xi}\cdot\ln H_0\tau)+c_{A2}\cos(\sqrt{\xi}\cdot\ln H_0\tau)\right](H_0\tau)^{\frac{1}{2}-\alpha} & w\in(w_;w_+) \\
		-\frac{1}{2\alpha+1}\left[c_{A1}(H_0\tau)^{\sqrt{-\xi}}+c_{A2}(H_0\tau)^{-\sqrt{-\xi}}\right](H_0\tau)^{\frac{1}{2}-\alpha} & w\not\in(w_;w_+)
	\end{cases} \\
	&s.t. \quad u_{C0}:=\frac{2(\alpha+1)u_{D0}-u_{A0}}{2\alpha+1}.
\end{align}

The evolution of $B$ is determined by Lemma \ref{lem3}.%
\begin{align}
\label{uB}
	&u_B(\tau)=H_0^{-3}\frac{u_{C0}-u_{A0}}{(n+1)\alpha+3}(H_0\tau)^{(n-1)\alpha+3}+u_{HB}(\tau) \cr
	&s.t. \quad u_{HB}(\tau)=a(\tau)^{-2}\int a(\tau')^2\left(\frac{2u_{HD}(\tau')-u_{HA}(\tau')}{2\alpha+1}-u_{HA}(\tau')\right)d\tau'.
\end{align}

\subsection{Particular components}

In the following sections, we will need the single-component solutions for some particular components.%
The dark energy has $w=-1<w_-$. The perturbations evolve as%
\begin{align}
\label{u de}
	H_0^2u_A(\tau)&=u_{A0}|_{w=-1}(H_0\tau-H_0c_{\Lambda})^3+c_{A1\Lambda}(H_0\tau-H_0c_{\Lambda})^2+c_{A2\Lambda}(H_0\tau-H_0c_{\Lambda}); \cr
	H_0^2u_C(\tau)&=u_{C0}|_{w=-1}(H_0\tau-H_0c_{\Lambda})^3-c_{A1\Lambda}(H_0\tau-H_0c_{\Lambda})^2-c_{A2\Lambda}(H_0\tau-H_0c_{\Lambda}) \cr
	&+c_{D1\Lambda}(H_0\tau-H_0c_{\Lambda})^5+c_{D2\Lambda}; \cr
	H_0^3u_B(\tau)&=-2c_{A1\Lambda}(H_0\tau-H_0c_{\Lambda})^3-2c_{A2\Lambda}(H_0\tau-H_0c_{\Lambda})^2\ln|H_0\tau-H_0c_{\Lambda}| \cr
	&+\frac{1}{4}c_{D1\Lambda}(H_0\tau-H_0c_{\Lambda})^6-c_{D2\Lambda}(H_0\tau-H_0c_{\Lambda})+c_{B\Lambda}(H_0\tau-H_0c_{\Lambda})^2.
\end{align}%
Indeed, $\alpha(-1)=-1\Rightarrow\sqrt{\xi}=\frac{1}{2}$ and $n(-1)=0$, thus%
\begin{equation}
	u_{A0}|_{w=-1}=u_{C0}|_{w=-1}=\frac{1}{2}, \Rightarrow u_{C0}-u_{A0}|_{w=-1}=0.
\end{equation}%
The matter has $w=0\in(w_-;w_+)$. The perturbations evolve as%
\begin{align}
\label{u mat}
	H_0^2u_A(\tau)&=u_{A0}|_{w=0}(H_0\tau-H_0c_M)^2+(H_0\tau-H_0c_M)^{-\frac{3}{2}}[c_{A1M}\sin\left(\frac{\sqrt{71}}{2}\ln(H_0\tau-H_0c_M)\right) \cr
	&+c_{A2M}\cos\left(\frac{\sqrt{71}}{2}\ln(H_0\tau-H_0c_M)\right)]; \cr
	H_0^2u_C(\tau)&=u_{C0}|_{w=0}(H_0\tau-H_0c_M)^2+c_{D1M}(H_0\tau-H_0c_M)^{-3}+c_{D2M}+\frac{1}{5}(H_0\tau-H_0c_M)^{-\frac{3}{2}}\cdot \cr
	&\cdot\left[c_{A1M}\sin\left(\frac{\sqrt{71}}{2}\ln(H_0\tau-H_0c_M)\right)+c_{A2M}\cos\left(\frac{\sqrt{71}}{2}\ln(H_0\tau-H_0c_M)\right)\right]; \cr
	H_0^3u_B(\tau)&=\frac{u_{C0}-u_{A0}}{7}|_{w=0}(H_0\tau-H_0c_M)^3+\frac{1}{2}c_{D1M}(H_0\tau-H_0c_M)^{-2}+\frac{1}{5}c_{D2M}(H_0\tau-H_0c_M) \cr
	&-\frac{1}{50}(H_0\tau-H_0c_M)^{-\frac{1}{2}}[(3c_{A1M}+\frac{\sqrt{71}}{2}c_{A2M})\sin\left(\frac{\sqrt{71}}{2}\ln(H_0\tau-H_0c_M)\right) \cr
	&+(3c_{A2M}-\sqrt{71}c_{A1M})\cos\left(\frac{\sqrt{71}}{2}\ln(H_0\tau-H_0c_M)\right)]+c_{BM}(H_0\tau-H_0c_M)^{-4}.
\end{align}%
Indeed, $\alpha(0)=2\Rightarrow\sqrt{-\xi}=\frac{\sqrt{71}}{2}$ and $n(0)=1$, thus%
\begin{equation}
	u_{A0}|_{w=0}=-\frac{1}{30}, \; \; \; , u_{C0}|_{w=0}=-\frac{17}{150} \Rightarrow \frac{u_{C0}-u_{A0}}{7}|_{w=0}=-\frac{2}{175}.
\end{equation}%
For the peculiar evolution during the radiation epoch, we don't use $T=a^n$. The perturbations
evolve as%
\begin{align}
\label{u rad}
	H_0^2u_A(\tau)&=\frac{1}{8}(H_0\tau)\left[\ln\left(\frac{H_0\tau}{4a_R}-\frac{3}{8}\right) \right]+u_{HA}(\tau); \cr
	H_0^2u_C(\tau)&=\frac{1}{8}(H_0\tau)\left[5\ln\left(\frac{H_0\tau}{4a_R}-\frac{63}{8}\right) \right]+u_{HC}(\tau); \cr
	H_0^3u_B(\tau)&=\frac{1}{8}(H_0\tau)^2\left[\ln\left(\frac{H_0\tau}{4a_R}-\frac{17}{8}\right) \right]+u_{HB}(\tau).
\end{align}

\subsection{Other Selfconsistence Conditions}

Recalling our definition of a \lq\lq selfconsistent\rq\rq universe, the First Selfconsistence
Condition ensures that there exist finite solutions for $\langle A\rangle$, $\langle B\rangle$, $\langle C\rangle$. We must require also
that these solutions are unique and that they describe small enough perturbations. The initial
conditions for (\ref{lin ein eq}) were%
\begin{equation}
\label{ic}
	\lim_{\tau\rightarrow0}\langle A\rangle(\tau), \langle B\rangle(\tau), \langle C\rangle(\tau)=0.
\end{equation}%
These functions are described by (\ref{uA}), (\ref{uB}) and (\ref{uC}) accordingly to the dominating $w$ near
$\tau=0$, i.e. $w_M$. The initial conditions put some restraints on $w_M$ and on the integration constants: we
can satisfy (\ref{ic}) if all $u_{HA}, u_{HB}, u_{HC}\equiv0$, and also%
\begin{align*}
		0=&\lim_{\tau\rightarrow0}(H_0\tau)^{(n(w_M)-1)\alpha(w_M)+2}\propto\lim_{\tau\rightarrow0} a(\tau)^{n(w_M)+3w_M} \Leftrightarrow n(w_M)+3w_M>0 \cr
		0=&\lim_{\tau\rightarrow0}(H_0\tau)^{(n(w_M)-1)\alpha(w_M)+3}\propto\lim_{\tau\rightarrow0} a(\tau)^{n(w_M)+3w_M+\frac{1}{\alpha(w_M)}} \cr
		&\Leftrightarrow n(w_M)+3w_M+\frac{1}{\alpha(w_M)}>0.
\end{align*}%
For the I SC one has $\alpha(w_M)>0$, so that the first limit implies the second one.

\begin{thm} [Second Selfconsistence Condition]
	A selfconsistent choice of components must be such that $n(w_M)+3w_M>0$.
\end{thm}

\begin{oss}
For a monotonically increasing $n(w)$, the II SC is equivalent to%
\begin{equation}
	w_M>w_0
\end{equation}%
for some limit value $w_0$. We can estimate it with a linear interpolation $n(w)\cong1+w$, that
gives $w_0\cong-\frac{1}{4}$. With more generality, remembering from \cite{re} that\footnote{$\phi\cong0.618...$ is the golden ratio.} $n(-\frac{1}{3})\in(\phi;1)$ and
that $n(0)=1$, we have%
\begin{equation}
	-\frac{1}{3}<w_0<0.
\end{equation}
\end{oss}

The II SC is not necessary if $w_M=\frac{1}{3}$, for which always%
\begin{equation}
	\lim_{\tau\rightarrow0}(H_0\tau)\left[\ln\left(\frac{H_0\tau}{4a_R}-\frac{3}{8}\right)\right]=0;
\end{equation}%
and the same for $B$ and $C$.

\begin{cor}
	The II SC ensures that the perturbations are small near the Big Bang.
\end{cor}
\begin{proof}
	\begin{equation}
		\frac{|\langle\tilde{g}_{\mu\nu}\rangle|}{|\bar{g}_{\mu\nu}|}\propto|\langle\Psi\rangle|=|\langle\Phi\rangle|=|\langle C\rangle-H\langle B\rangle|\rightarrow^{\tau\rightarrow0}0\ll1;
	\end{equation}%
	where the second term is $H\langle B\rangle\propto\frac{1}{\tau}(H_0\tau)^{(n(w_M)-1)\alpha(w_M)+3}\propto a^{n(w_M)+3w_M}\rightarrow^{\tau\rightarrow0}0$ for
the II SC again.
\end{proof}

Do the initial conditions (\ref{ic}) fix uniquely $\langle A\rangle$, $\langle B\rangle$, $\langle C\rangle$? Not always. There are values
of $w$ for which $u_{HA}, u_{HB}, u_{HC}$ go to zero even if the integration constants are not fixed
to zero. Such cases are not selfconsistent, because the solutions are not unique. This is
forbidden by

\begin{thm} [Third Selfconsistence Condition]
	A selfconsistent choice of components must be such that $w_-<w_M\leq1$.
\end{thm}
\begin{proof}
	Let us try any non zero choice for the integration constant, and check if nevertheless $u_{HA}$ tends
to zero; if it is the case, the corresponding value of $w$ will not be selfconsistent.%
	First, let us consider the case $w_M\in(w_-;w_+)$. Remembering (\ref{uA})%
	\begin{equation*}
		\lim_{\tau\rightarrow0}u_{HA}(\tau)=0 \Leftrightarrow \lim_{\tau\rightarrow0}(H_0\tau)^{\frac{1}{2}-\alpha(w_M)}=0 \Leftrightarrow \alpha(w_M)<\frac{1}{2} \Leftrightarrow w_M>1,
	\end{equation*}%
	otherwise the limit doesn't exixt because of oscillations. This forbids the values $w_M\in(1;w_+)$.%
	Considering now the case $w_M\not\in(w_-;w_+)$, it means $\alpha(w_M)\in\left[-\frac{1}{\sqrt{3}}-\frac{1}{2};\frac{1}{\sqrt{3}}-\frac{1}{2}\right]$, and in particular
$\alpha<\frac{1}{2}$. Recalling (\ref{uA}), for a choice $c_{A1}\neq0$, $c_{A2}=0$%
	\begin{align*}
		&\sqrt{\frac{1}{4}-3\alpha-3\alpha^2}+\frac{1}{2}-\alpha|_{w_M}>\frac{1}{2}-\alpha(w_M)>0 \cr
		&\Rightarrow \lim_{\tau\rightarrow0}(H_0\tau)^{\sqrt{\frac{1}{4}-3\alpha-3\alpha^2}+\frac{1}{2}-\alpha}|_{w_M}=0 \cr
		&\Rightarrow \lim_{\tau\rightarrow0}u_{HA}(\tau)=0,
	\end{align*}%
	and this is enough to forbid all the values $w_M\not\in(w_-;w_+)$.%
	For the allowed values $w_M\in(w_-;1]$, the integration constants for $u_C$ are fixed to zero
as well, since (\ref{uC}) has the same functional form of (\ref{uA}). From (\ref{uB}) and (\ref{ic}) we see that
also $u_{HB}$ is fixed to zero, so that the metric perturbations are unique.
\end{proof}

The Three Selfconsistence Conditions we proved allow only a \lq\lq selfconsistence interval\rq\rq~for the component dominating near the Big Bang:%
\begin{equation}
	w_M\in(w_0;1].
\end{equation}

\begin{oss}
	Our universe contains certainly radiation and matter as homogeneous
components, and probably dark energy. The biggest $w$ is that of radiation, and $w_M=\frac{1}{3}$
is included in the selfconsistence interval. This is not obvious. Some universes, as the
\lq\lq constant coefficient universe\rq\rq~studied in \cite{re}, break the Three Selfconsistence Conditions.
The selfconsistence of our universe provides some empirical reinforcement to our model.
\end{oss}

When I and III SC hold, the requirement of selfconsistence is reduced to asking
that the perturbations are small enough to neglect orders higher than the first. This
constitues a last Condition.

\begin{lem} [Fourth Selfconsistence Condition]
	A selfconsistent choice of components
must have an inhomogeneous matter such that $\Omega_{IM0}\ll\Omega_{TM0}$, and such that $\forall t\in[0;t_0]:|\langle\Psi\rangle|\ll\frac{1}{2}$.
\end{lem}
\begin{proof}
	The first requirement on $\Omega_{IM0}$ is the same we asked in \S2.4. The other requirement
is evident from (\ref{aver g}), where $2\langle\Psi\rangle$, $2\langle\Phi\rangle$ are the perturbations of the metric. They must
be smaller than $1$, and, since $\Psi=\Phi$, it is sufficient to impose it just for one of them.
\end{proof}

This is no more a Condition on $w$, but on $\Omega_{IM0}$, so that the selfconsistence interval
remains the same. Indeed, $\langle A\rangle$, $\langle B\rangle$, $\langle C\rangle$ are proportional to $\Omega_{IM0}$, and so are $\langle\Psi\rangle$, $\langle\Phi\rangle$:
the IV SC defines a maximum value $\Omega^M_{IM0}$ for the inhomogeneity.

\begin{oss}
	Notice that the IV SC does not imply the II SC, since in the limit case
$w_M=w_0$, $\langle\Psi\rangle$ doesn't tend to zero, but, nevertheless, it could be small.
\end{oss}

\section{A model for the real universe}

\subsection{The 1-manifold of possible universes}

Until now, our computations concerned a general choice of components $\{\bar{\Omega}_{w0}\}_w$, for which
we found the Selfconsistence Conditions. Now we will apply this general method to our
universe.

It contains just three components: radiation $\Omega_{R0}$, matter $\Omega_{M0}$ and dark energy $\Omega_{\Lambda0}$.
These are fixed by the measures of $\Omega_{R0}$, of $q_0=\Omega_{R0}+\frac{1}{2}\Omega_{M0}-\Omega_{\Lambda0}$ and of the space flatness \cite{Spergel:2003cb}
$1=1-\Omega_{k0}=\Omega_{R0}+\Omega_{M0}+\Omega_{\Lambda0}$. The background components are as well $\bar{\Omega}_{R0}$, $\bar{\Omega}_{M0}$ and
$\bar{\Lambda0}$, on which the model puts the restraints%
\begin{equation}
	\begin{cases}
		\left(\frac{H_0}{\bold{H}_0}\tilde{a}_0\right)^2\bar{\Omega}_{R0}=\Omega_{R0} \\

		\Omega_{FM0}+\left(\frac{H_0}{\bold{H}_0}\tilde{a}_0\right)^2\bar{\Omega}_{M0}+\Omega_{IM0}=\Omega_{M0} \\

		\Omega_{F\Lambda0}+\left(\frac{H_0}{\bold{H}_0}\tilde{a}_0\right)^2\bar{\Omega}_{\Lambda0}=\Omega_{\Lambda0}
	\end{cases}.
\end{equation}%
Notice that these are not independent, since $\bar{\Omega}_{R0}+\bar{\Omega}_{M0}+\bar{\Omega}_{\Lambda0}:=1$. We have only two
independent constraints from%
\begin{equation}
\label{restr}
	\begin{cases}
		[1-(sum+1)\Omega_{IM0}]\bar{\Omega}_{R0}=\Omega_{R0} \\

		\frac{2sum+ract}{3}\Omega_{IM0}+[1-(sum+1)\Omega_{IM0}]\bar{\Omega}_{M0}+\Omega_{IM0}=\Omega_{M0} \\

		\frac{sum-ract}{3}\Omega_{IM0}+[1-(sum+1)\Omega_{IM0}]\bar{\Omega}_{\Lambda0}=\Omega_{\Lambda0}=1-\Omega_{R0}-\Omega_{M0}
	\end{cases}.
\end{equation}%
However, we have three unknown parameters: the inhomogeneity $\Omega_{IM0}$ and other two
among $\bar{\Omega}_{R0}$, $\bar{\Omega}_{M0}$ and $\bar{\Omega}_{\Lambda0}$. This means that the components of our universe are not
completely determined by (\ref{restr}), but we will find more possible solutions, when a parameter
changes. We choose $\bar{\Omega}_{M0}\in[0;1]$ as parameter, with $\Omega_{\Lambda0}(\Omega_{R0}; \Omega_{M0})=1-\Omega_{R0}-\Omega_{M0}$,%
\begin{equation}
	\Omega_{IM0}(\Omega_{R0}; \Omega_{M0})=\frac{1}{1+sum(\Omega_{R0}; \Omega_{M0})}\left(1-\frac{\Omega_{R0}}{\bar{\Omega}_{R0}}\right)
\end{equation}%
and $\bar{\Omega}_{R0}=\bar{\Omega}_{R0}(\bar{\Omega}_{M0})$ is determined by the last independent constraint of (\ref{restr}).

We will have to check which of these values of $\bar{\Omega}_{M0}$ gives
selfconsistent (i.e., if for them hold the IV SC and that $\Omega_{IM0}\ll\Omega_{TM0}$), acceptable
and evetually good solutions.

\subsection{Epochs of evolution}

Applying (\ref{bg comp}),%
\begin{equation*}
		\bar{\Omega}_R=\bar{\Omega}_{R0}a^{-4}, \quad \bar{\Omega}_M=\bar{\Omega}_{R0}a^{-3}, \quad \bar{\Omega}_{\Lambda}\equiv\bar{\Omega}_{\Lambda0}.
\end{equation*}%
So we can get the values of $a$ for which the matter starts to be more than the radiation,
and the same for other couples%
\begin{align}
		\bar{\Omega}_R\geq\bar{\Omega}_M &\Leftrightarrow a\leq a_{RM}:=a(\tau_{RM})=\frac{\bar{\Omega}_{R0}}{\bar{\Omega}_{M0}} \cr
		\bar{\Omega}_M\geq\bar{\Omega}_{\Lambda} &\Leftrightarrow a\leq a_{M\Lambda}:=a(\tau_{M\Lambda})=\sqrt[3]{\frac{\bar{\Omega}_{M0}}{\bar{\Omega}_{\Lambda0}}} \cr
		\bar{\Omega}_R\geq\bar{\Omega}_{\Lambda} &\Leftrightarrow a\leq a_{R\Lambda}:=a(\tau_{R\Lambda})=\sqrt[4]{\frac{\bar{\Omega}_{R0}}{\bar{\Omega}_{\Lambda0}}}.
\end{align}

The evolution of the universe until now is for $0\leq a\leq a_0\cong1$. During this time, there
may have been three or two epochs, depending on the values $\bar{\Omega}_{R0}$, $\bar{\Omega}_{M0}$ and $\bar{\Omega}_{\Lambda0}$.

\begin{lem}
\label{epoch lem}
	A selfconsistent background evolution can be divided in epochs in the following ways.
	\begin{itemize}
		\item If $a_{RM}<a_{M\Lambda}<a_0$, then there are three epochs: for radiation $[0;\tau_{RM}]$, matter
$[\tau_{RM};\tau_{M\Lambda}]$, and dark energy $[\tau_{M\Lambda};\tau(t_0)]$.

		\item If the first inequality does not hold, then there are just two epochs: radiation $[0;\tau_{R\Lambda}]$
and dark energy $[\tau_{R\Lambda};\tau(t_0)]$.

		\item If the second inequality does not hold, then there are just two epochs: radiation
$[0;\tau_{RM}]$ and matter $[\tau_{RM};\tau(t_0)]$.
	\end{itemize}
\end{lem}
\begin{proof}
	Since the radiation exists, we know $\bar{\Omega}_{R0}>0$, so that $a_{RM}>0$ and $a_{R\Lambda}>0$ for any
values of $\bar{\Omega}_{M0}$, $\bar{\Omega}_{\Lambda0}$. Thus, we have always a radiation epoch, which is the first one
after the Big Bang. The presence of other epochs depends on our parameter: the quantity of homogeneous matter $\bar{\Omega}_{M0}$.

	We will not consider the case with only the radiation epoch, because it would mean
that $a_{RM}, a_{R\Lambda}\geq a_0\cong1$, which appens for high values of $\bar{\Omega}_{R0}$; but we know from the measures \cite{Spergel:2003cb}
that the radiation is far more less than the matter. Moreover, if the homogeneous matter
would be so little, it would mean that $\Omega_{IM0}\cong\Omega_{TM0}$, that is not selfconsistent.
\end{proof}

Let us consider the two cases with a matter epoch. From (\ref{single a}) we get the background evolution%
\begin{equation}
	a(\tau)=\begin{cases}
		H_0\tau & \tau\in[0;\tau_{RM}] \\

		\frac{H_0^2}{4}(\tau-c_M)^2 & \tau\in[\tau_{RM};\tau_{M\Lambda}] \\

		\frac{1}{H_0(c_{\Lambda}-\tau)} & \tau\in[\tau_{M\Lambda};\tau_F]
	\end{cases},
\end{equation}%
where the continuity determines%
\begin{align}
		H_0\tau_{RM}=a_{RM} &\Rightarrow H_0c_M=a_{RM}-2\sqrt{a_{RM}}, \cr
		H_0\tau_{M\Lambda}=2\sqrt{a_{M\Lambda}}+H_0c_M,& \quad H_0c_{\Lambda}=H_0\tau_{M\Lambda}+4(H_0\tau_{M\Lambda}-H_0c_M)^{-2}.
\end{align}

On the other hand, in the case such that there is no matter epoch, the evolution is%
\begin{equation}
	a(\tau)=\begin{cases}
		H_0\tau & \tau\in[0;\tau_{R\Lambda}] \\

		\frac{1}{H_0(c_R-\tau)} & \tau\in[\tau_{R\Lambda};\tau_F]
	\end{cases},
\end{equation}%
where the continuity determines%
\begin{equation}
	H_0\tau_{R\Lambda}=a_{R\Lambda} \Rightarrow H_0c_R=a_{R\Lambda}+\frac{1}{a_{R\Lambda}}.
\end{equation}

\subsection{Three epochs}

From the results of previous sections, now we can get the formulas for the evaluation of
fictitious dark matter and dark energy. They are different for the three cases described in
Lemma \ref{epoch lem}. Let us start with the case where we have all the three epochs.

Applying (\ref{u rad}), (\ref{u mat}) and (\ref{u de}), we get the evolution of $\langle A\rangle$%
\begin{equation}
	H_0^2u_A(\tau)=\begin{cases}
		\frac{1}{8}(H_0\tau)\left[\ln\left(\frac{H_0\tau}{4a_{RM}}\right)-\frac{3}{8}\right] & \tau\in[0;\tau_{RM}] \\

		-\frac{1}{30}(H_0\tau-H_0c_M)^2 \\
		\quad +(H_0\tau-H_0c_M)^{-\frac{3}{2}}[c_{A1M}\sin\left(\frac{\sqrt{71}}{2}\ln(H_0\tau-H_0c_M)\right) \\
		\quad +c_{A2M}\cos\left(\frac{\sqrt{71}}{2}\ln(H_0\tau-H_0c_M)\right)] & \tau\in[\tau_{RM};\tau_{M\Lambda}] \\

		\frac{1}{2}(H_0\tau-H_0c_{\Lambda})^3+c_{A1\Lambda}(H_0\tau-H_0c_{\Lambda})^2 \\
		\quad +c_{A2\Lambda}(H_0\tau-H_0c_{\Lambda}) & \tau\in[\tau_{M\Lambda};\tau(t_0)]
	\end{cases}
\end{equation}%
where the $C^1$ regularity fixes the integration constants $c_{A1M}$, $c_{A2M}$ s.t.%
\begin{equation}
\begin{cases}
	-\frac{1}{8}\left(\frac{3}{8}+\ln4\right)(H_0\tau_{RM})=-\frac{1}{30}(H_0\tau_{RM}-H_0c_M)^2 \\
	\quad +(H_0\tau_{RM}-H_0c_M)^{-\frac{3}{2}}[c_{A1M}\sin\left(\frac{\sqrt{71}}{2}\ln(H_0\tau_{RM}-H_0c_M)\right) \\
	\quad +c_{A2M}\cos\left(\frac{\sqrt{71}}{2}\ln(H_0\tau_{RM}-H_0c_M)\right)] \\

	\frac{1}{8}\left(\frac{5}{8}-\ln4\right)=-\frac{1}{15}(H_0\tau_{RM}-H_0c_M)+(H_0\tau_{RM}-H_0c_M)^{-\frac{5}{2}}\cdot \\
	\quad \cdot[c_{A1M}\left[-\frac{3}{2}\sin\left(\frac{\sqrt{71}}{2}\ln(H_0\tau_{RM}-H_0c_M)\right)+\frac{\sqrt{71}}{2}\cos\left(\frac{\sqrt{71}}{2}\ln(H_0\tau_{RM}-H_0c_M)\right)\right] \\
	\quad +c_{A2M}\left[-\frac{3}{2}\cos\left(\frac{\sqrt{71}}{2}\ln(H_0\tau_{RM}-H_0c_M)\right)-\frac{\sqrt{71}}{2}\sin\left(\frac{\sqrt{71}}{2}\ln(H_0\tau_{RM}-H_0c_M)\right)\right]]
\end{cases}
\end{equation}%
(the left hand terms have been simplified using $a_{RM}=H_0\tau_{RM}$), and $c_{A1\Lambda}$, $c_{A2\Lambda}$ s.t.%
\begin{equation}
\begin{cases}
	-\frac{1}{30}(H_0\tau_{M\Lambda}-H_0c_M)^2+(H_0\tau_{M\Lambda}-H_0c_M)^{-\frac{3}{2}}\cdot \\
	\quad \cdot\left[c_{A1M}\sin\left(\frac{\sqrt{71}}{2}\ln(H_0\tau_{M\Lambda}-H_0c_M)\right)+c_{A2M}\cos\left(\frac{\sqrt{71}}{2}\ln(H_0\tau_{M\Lambda}-H_0c_M)\right)\right] \\
	\quad =\frac{1}{2}(H_0\tau_{M\Lambda}-H_0c_{\Lambda})^3+c_{A1\Lambda}(H_0\tau_{M\Lambda}-H_0c_{\Lambda})^2+c_{A2\Lambda}(H_0\tau_{M\Lambda}-H_0c_{\Lambda}) \\

	-\frac{1}{15}(H_0\tau_{M\Lambda}-H_0c_M)-\frac{1}{2}(H_0\tau_{M\Lambda}-H_0c_M)^{-\frac{5}{2}}\cdot \\
	\quad \cdot[(3c_{A1M}+\sqrt{71}c_{A2M})\sin\left(\frac{\sqrt{71}}{2}\ln(H_0\tau_{M\Lambda}-H_0c_M)\right) \\
	\quad +(3c_{A2M}-\sqrt{71}c_{A1M})\cos\left(\frac{\sqrt{71}}{2}\ln(H_0\tau_{M\Lambda}-H_0c_M)\right)] \\
	\quad =\frac{3}{2}(H_0\tau_{M\Lambda}-H_0c_{\Lambda})^2+2c_{A1\Lambda}(H_0\tau_{M\Lambda}-H_0c_{\Lambda})+c_{A2\Lambda}
\end{cases}
\end{equation}%
the evolution of $\langle C\rangle$%
\begin{equation}
H_0^2u_C(\tau)=\begin{cases}
	\frac{1}{8}(H_0\tau)\left[5\ln\left(\frac{H_0\tau}{4a_{RM}}\right)-\frac{63}{8}\right] & \tau\in[0;\tau_{RM}] \\

	-\frac{17}{150}(H_0\tau-H_0c_M)^2+c_{D1M}(H_0\tau-H_0c_M)^{-3}+c_{D2M} \\
	\quad \frac{1}{5}(H_0\tau-H_0c_M)^{-\frac{3}{2}}[c_{A1M}\sin\left(\frac{\sqrt{71}}{2}\ln(H_0\tau-H_0c_M)\right) \\
	\quad +c_{A2M}\cos\left(\frac{\sqrt{71}}{2}\ln(H_0\tau-H_0c_M)\right)] & \tau\in[\tau_{RM};\tau_{M\Lambda}] \\

	\frac{1}{2}(H_0\tau-H_0c_{\Lambda})^3-c_{A1\Lambda}(H_0\tau-H_0c_{\Lambda})^2 \\
	\quad -c_{A2\Lambda}(H_0\tau-H_0c_{\Lambda})+c_{D1\Lambda}(H_0\tau-H_0c_{\Lambda})^5+c_{D2\Lambda} & \tau\in[\tau_{M\Lambda};\tau(t_0)]
\end{cases}
\end{equation}%
where the $C^1$ regularity fixes the integration constants $c_{D1M}$, $c_{D2M}$ s.t.%
\begin{equation}
\begin{cases}
	-\frac{1}{8}\left(\frac{63}{8}+5\ln4\right)(H_0\tau_{RM})=-\frac{17}{150}(H_0\tau_{RM}-H_0c_M)^2 \\
	\quad +c_{D1M}(H_0\tau_{RM}-H_0c_M)^{-3}+c_{D2M}+\frac{1}{5}(H_0\tau_{RM}-H_0c_M)^{-\frac{3}{2}}\cdot \\
	\quad \cdot\left[c_{A1M}\sin\left(\frac{\sqrt{71}}{2}\ln(H_0\tau_{RM}-H_0c_M)\right)+c_{A2M}\cos\left(\frac{\sqrt{71}}{2}\ln(H_0\tau_{RM}-H_0c_M)\right)\right] \\

	-\frac{1}{8}\left(\frac{23}{8}+5\ln4\right)=-\frac{17}{75}(H_0\tau_{RM}-H_0c_M)-3c_{D1M}(H_0\tau_{RM}-H_0c_M)^{-4} \\
	\quad -\frac{1}{10}(H_0\tau_{RM}-H_0c_M)^{-\frac{5}{2}}[(3c_{A1M}+\sqrt{71}c_{A2M})\sin\left(\frac{\sqrt{71}}{2}\ln(H_0\tau_{RM}-H_0c_M)\right) \\
	\quad +(3c_{A2M}-\sqrt{71}c_{A1M})\cos\left(\frac{\sqrt{71}}{2}\ln(H_0\tau_{RM}-H_0c_M)\right)]
\end{cases}
\end{equation}%
and $c_{D1\Lambda}$, $c_{D2\Lambda}$ s.t.%
\begin{equation}
\begin{cases}
	-\frac{17}{150}(H_0\tau_{M\Lambda}-H_0c_M)^2+\frac{1}{5}(H_0\tau_{M\Lambda}-H_0c_M)^{-\frac{3}{2}}[c_{A1M}\sin\left(\frac{\sqrt{71}}{2}\ln(H_0\tau_{M\Lambda}-H_0c_M)\right) \\
	\quad +c_{A2M}\cos\left(\frac{\sqrt{71}}{2}\ln(H_0\tau_{M\Lambda}-H_0c_M)\right)]+c_{D1M}(H_0\tau_{M\Lambda}-H_0c_M)^{-3}+c_{D2M} \\
	\quad =\frac{1}{2}(H_0\tau_{M\Lambda}-H_0c_{\Lambda})^3-c_{A1\Lambda}(H_0\tau_{M\Lambda}-H_0c_{\Lambda})^2-c_{A2\Lambda}(H_0\tau_{M\Lambda}-H_0c_{\Lambda}) \\
	\quad +c_{D1\Lambda}(H_0\tau_{M\Lambda}-H_0c_{\Lambda})^5+c_{D2\Lambda} \\

	-\frac{17}{75}(H_0\tau_{M\Lambda}-H_0c_M)-3c_{D1M}(H_0\tau_{M\Lambda}-H_0c_M)^{-4}-\frac{1}{10}(H_0\tau_{M\Lambda}-H_0c_M)^{-\frac{5}{2}}\cdot \\
	\quad \cdot[(3c_{A1M}+\sqrt{71}c_{A2M})\sin\left(\frac{\sqrt{71}}{2}\ln(H_0\tau_{M\Lambda}-H_0c_M)\right) \\
	\quad +(3c_{A2M}-\sqrt{71}c_{A1M})\cos\left(\frac{\sqrt{71}}{2}\ln(H_0\tau_{M\Lambda}-H_0c_M)\right)]=\frac{3}{2}(H_0\tau_{M\Lambda}-H_0c_{\Lambda})^2 \\
	\quad -2c_{A1\Lambda}(H_0\tau_{M\Lambda}-H_0c_{\Lambda})-c_{A2\Lambda}+5c_{D1\Lambda}(H_0\tau_{M\Lambda}-H_0c_{\Lambda})^4
\end{cases}
\end{equation}%
and the evolution of $\langle B\rangle$%
\begin{equation}
	H_0^3u_B(\tau)=\begin{cases}
		\frac{1}{8}(H_0\tau)^2\left[\ln\left(\frac{H_0\tau}{4a_{RM}}\right)-\frac{17}{8}\right] & \tau\in[0;\tau_{RM}] \\

		-\frac{2}{175}(H_0\tau-H_0c_M)^3+\frac{1}{2}c_{D1M}(H_0\tau-H_0c_M)^{-2} \\
		\quad +\frac{1}{5}c_{D2M}(H_0\tau-H_0c_M)-\frac{1}{50}(H_0\tau-H_0c_M)^{-\frac{1}{2}}\cdot \\
		\quad \cdot[(3c_{A1M}+\frac{\sqrt{71}}{2}c_{A2M})\sin\left(\frac{\sqrt{71}}{2}\ln(H_0\tau-H_0c_M)\right) \\
		\quad +(3c_{A2M}-\sqrt{71}c_{A1M})\cos\left(\frac{\sqrt{71}}{2}\ln(H_0\tau-H_0c_M)\right)] \\
		\quad +c_{BM}(H_0\tau-H_0c_M)^{-4} & \tau\in[\tau_{RM};\tau_{M\Lambda}] \\

		-2c_{A1\Lambda}(H_0\tau-H_0c_{\Lambda})^3-2c_{A2\Lambda}(H_0\tau-H_0c_{\Lambda})^2\ln|H_0\tau-H_0c_{\Lambda}| \\
		\; \; +\frac{1}{4}c_{D1\Lambda}(H_0\tau-H_0c_{\Lambda})^6-c_{D2\Lambda}(H_0\tau-H_0c_{\Lambda})+c_{B\Lambda}(H_0\tau-H_0c_{\Lambda})^2 & \tau\in[\tau_{M\Lambda};\tau(t_0)]
	\end{cases}.
\end{equation}%
where the continuity fixes the integration constants $c_{BM}$, $c_{B\Lambda}$ s.t.%
\begin{equation}
\begin{cases}
	-\frac{1}{8}\left(\frac{17}{8}+\ln4\right)(H_0\tau_{RM})^2=-\frac{2}{175}(H_0\tau_{RM}-H_0c_M)^3+\frac{1}{2}c_{D1M}(H_0\tau_{RM}-H_0c_M)^{-2} \\
	\quad +\frac{1}{5}c_{D2M}(H_0\tau_{RM}-H_0c_M)-\frac{1}{50}(H_0\tau_{RM}-H_0c_M)^{-\frac{1}{2}}\cdot \\
	\quad \cdot[(3c_{A1M}+\sqrt{71}c_{A2M})\sin\left(\frac{\sqrt{71}}{2}\ln(H_0\tau_{RM}-H_0c_M)\right) \\
	\quad +(3c_{A2M}-\sqrt{71}c_{A1M})\cos\left(\frac{\sqrt{71}}{2}\ln(H_0\tau_{RM}-H_0c_M)\right)]+c_{BM}(H_0\tau_{RM}-H_0c_M)^{-4} \\

	-\frac{2}{175}(H_0\tau_{M\Lambda}-H_0c_M)^3+\frac{1}{2}c_{D1M}(H_0\tau_{M\Lambda}-H_0c_M)^{-2}+\frac{1}{5}c_{D2M}(H_0\tau_{M\Lambda}-H_0c_M) \\
	\quad -\frac{1}{50}(H_0\tau_{M\Lambda}-H_0c_M)^{-\frac{1}{2}}[(3c_{A1M}+\sqrt{71}c_{A2M})\sin\left(\frac{\sqrt{71}}{2}\ln(H_0\tau_{M\Lambda}-H_0c_M)\right) \\
	\quad +(3c_{A2M}-\sqrt{71}c_{A1M})\cos\left(\frac{\sqrt{71}}{2}\ln(H_0\tau_{M\Lambda}-H_0c_M)\right)]+c_{BM}(H_0\tau_{M\Lambda}-H_0c_M)^{-4} \\
	\quad =-2c_{A1\Lambda}(H_0\tau_{M\Lambda}-H_0c_{\Lambda})^{3}-2c_{A2\Lambda}(H_0\tau_{M\Lambda}-H_0c_{\Lambda})^{2}\ln|H_0\tau_{M\Lambda}-H_0c_{\Lambda}| \\
	\quad +\frac{1}{4}c_{D1\Lambda}(H_0\tau_{M\Lambda}-H_0c_{\Lambda})^{6}-c_{D2\Lambda}(H_0\tau_{M\Lambda}-H_0c_{\Lambda})+c_{B\Lambda}(H_0\tau_{M\Lambda}-H_0c_{\Lambda})^{2}
\end{cases}
\end{equation}

Now, in order to get the fictitious components $\Omega_{FM0}, \Omega_{F\Lambda0}$ we need only to apply the formulas
(\ref{sum ract}). Recalling%
\begin{equation}
	\frac{H'_0}{H_0}=-\bar{q}_0=-\frac{1}{2}\sum_w(1+3w)\bar{\Omega}_{w0}=-\bar{\Omega}_{R0}-\frac{1}{2}\bar{\Omega}_{M0}+\bar{\Omega}_{\Lambda0}=1-2\bar{\Omega}_{R0}-\frac{3}{2}\bar{\Omega}_{M0}.
\end{equation}%
In a similar way, we can evaluate the term $\frac{(H'_0)^2}{H_0^2}+\frac{H''_0}{H_0}$ from (\ref{bg evol})%
\begin{equation*}
	H(a)^2=H_0^2\sum_w\bar{\Omega}_{w0}a^{-1-3w} \Rightarrow 2HH''+2(H')^2=H_0^2\sum_w(-1-3w)(-2-3w)a^{-3-3w} \Rightarrow
\end{equation*}
\begin{equation}
	\frac{(H'_0)^2}{H_0^2}+\frac{H''_0}{H_0}=\frac{1}{2}\sum_w(1+3w)(2+3w)\bar{\Omega}_{w0}=3\bar{\Omega}_{R0}+\bar{\Omega}_{M0}+\bar{\Omega}_{\Lambda0}=1+2\bar{\Omega}_{R0}.
\end{equation}%
Thus, (\ref{sum ract}) become%
\begin{equation}
\label{sumract}
	\begin{cases}
		\frac{1}{2}(sum+1)\Omega_{IM0}=-\langle A\rangle_0-(2\bar{\Omega}_{R0}+\frac{3}{2}\bar{\Omega}_{M0})H_0\langle B\rangle_0-\langle C\rangle'_0 \\

		\frac{1}{2}(ract+1)\Omega_{IM0}=\langle A\rangle'_0+2\langle C\rangle'_0+\langle C\rangle''_0+(1-6\bar{\Omega}_{R0}-3\bar{\Omega}_{M0})H_0\langle B\rangle_0 \\
		\quad +(1-2\bar{\Omega}_{R0}-\frac{3}{2}\bar{\Omega}_{M0})(2\langle A\rangle_0-\langle C\rangle_0+\langle C\rangle'_0)
	\end{cases}.
\end{equation}%
Here, all perturbations are evaluated today, when the dark energy dominates%
\begin{equation*}
	a(\tau)=\frac{1}{H_0(c_{\Lambda}-\tau)} \Rightarrow H_0\tau-H_0c_{\Lambda}=-a^{-1}
\end{equation*}

\begin{equation*}
	\langle A\rangle=\frac{3}{2}\Omega_{IM0}[\frac{1}{2}(-a^{-1})^3+c_{A1\Lambda}(-a^{-1})^2+c_{A2\Lambda}(-a^{-1})] \Rightarrow
\end{equation*}
\begin{equation}
\label{A0 de}
	\langle A\rangle_0=\frac{3}{2}\Omega_{IM0}[-\frac{1}{2}+c_{A1\Lambda}-c_{A2\Lambda}], \; \; \; \langle A\rangle'_0=\frac{3}{2}\Omega_{IM0}[\frac{3}{2}-2c_{A1\Lambda}+c_{A2\Lambda}];
\end{equation}

\begin{align*}
	\langle B\rangle&=\frac{3}{2}\frac{\Omega_{IM0}}{H_0}[-2c_{A1\Lambda}(-a^{-1})^{3}-2c_{A2\Lambda}(-a^{-1})^{2}\ln|-a^{-1}| \cr
	&+\frac{1}{4}c_{D1\Lambda}(-a^{-1})^{6}-c_{D2\Lambda}(-a^{-1})+c_{B\Lambda}(-a^{-1})^{2}] \Rightarrow
\end{align*}
\begin{equation}
\label{B0 de}
	H_0\langle B\rangle_0=\frac{3}{2}\Omega_{IM0}[2c_{A1\Lambda}+\frac{1}{4}c_{D1\Lambda}+c_{D2\Lambda}+c_{B\Lambda}];
\end{equation}

\begin{equation*}
	\langle C\rangle=\frac{3}{2}\Omega_{IM0}[\frac{1}{2}(-a^{-1})^3-c_{A1\Lambda}(-a^{-1})^2-c_{A2\Lambda}(-a^{-1})+c_{D1\Lambda}(-a^{-1})^5+c_{D2\Lambda}] \Rightarrow
\end{equation*}
\begin{align}
\label{C0 de}
	\langle C\rangle_0&=\frac{3}{2}\Omega_{IM0}[-\frac{1}{2}-c_{A1\Lambda}+c_{A2\Lambda}-c_{D1\Lambda}+c_{D2\Lambda}], \cr
	\langle C\rangle'_0&=\frac{3}{2}\Omega_{IM0}[\frac{3}{2}+2c_{A1\Lambda}-c_{A2\Lambda}+5c_{D1\Lambda}], \cr
	\langle C\rangle''_0&=\frac{3}{2}\Omega_{IM0}[-6-6c_{A1\Lambda}+2c_{A2\Lambda}-30c_{D1\Lambda}].
\end{align}%
Replacing (\ref{A0 de}), (\ref{B0 de}) and (\ref{C0 de}) inside (\ref{sumract}), we obtain \emph{ract} and \emph{sum}.

\subsection{No matter epoch}

For different values of $\bar{\Omega}_{R0}, \bar{\Omega}_{M0}$, we would have only the radiation and dark energy epochs.
Applying (\ref{u rad}), (\ref{u mat}) and (\ref{u de}), we get the evolution of $\langle A\rangle$%
\begin{equation}
	H_0^2u_A(\tau)=\begin{cases}
		\frac{1}{8}(H_0\tau)\left[\ln\left(\frac{H_0\tau}{4a_{RM}}\right)-\frac{3}{8}\right] & \tau\in[0;\tau_{R\Lambda}] \\

		\frac{1}{2}(H_0\tau-H_0c_R)^3+c_{A1\Lambda}(H_0\tau-H_0c_R)^2 \\
		\quad +c_{A2\Lambda}(H_0\tau-H_0c_R) & \tau\in[\tau_{R\Lambda};\tau(t_0)]
	\end{cases}
\end{equation}%
where the $C^1$ regularity fixes the integration constants $c_{A1\Lambda}$, $c_{A2\Lambda}$ s.t.%
\begin{equation}
\begin{cases}
	-\frac{1}{8}\left(\frac{3}{8}+\ln4\right)(H_0\tau_{R\Lambda})=\frac{1}{2}(H_0\tau_{R\Lambda}-H_0c_R)^3 \\
	\quad +c_{A1\Lambda}(H_0\tau_{R\Lambda}-H_0c_R)^2+c_{A2\Lambda}(H_0\tau_{R\Lambda}-H_0c_R) \\

	\frac{1}{8}\left(\frac{5}{8}-\ln4\right)=\frac{3}{2}(H_0\tau_{R\Lambda}-H_0c_R)^2+2c_{A1\Lambda}(H_0\tau_{R\Lambda}-H_0c_R)+c_{A2\Lambda}
\end{cases}
\end{equation}%
For the evolution of $\langle C\rangle$ we get%
\begin{equation}
H_0^2u_C(\tau)=\begin{cases}
	\frac{1}{8}(H_0\tau)\left[5\ln\left(\frac{H_0\tau}{4a_{RM}}\right)-\frac{63}{8}\right] & \tau\in[0;\tau_{RM\Lambda}] \\

	\frac{1}{2}(H_0\tau-H_0c_R)^3-c_{A1\Lambda}(H_0\tau-H_0c_R)^2 \\
	\quad -c_{A2\Lambda}(H_0\tau-H_0c_R)+c_{D1\Lambda}(H_0\tau-H_0c_R)^5+c_{D2\Lambda} & \tau\in[\tau_{R\Lambda};\tau(t_0)]
\end{cases}
\end{equation}%
where the $C^1$ regularity fixes the integration constants $c_{D1\Lambda}$, $c_{D2\Lambda}$ s.t.%
\begin{equation}
\begin{cases}
	-\frac{1}{8}\left(\frac{63}{8}+5\ln4\right)(H_0\tau_{R\Lambda})=\frac{1}{2}(H_0\tau_{R\Lambda}-H_0c_R)^3-c_{A1\Lambda}(H_0\tau_{R\Lambda}-H_0c_R)^2 \\
	\quad -c_{A2\Lambda}(H_0\tau_{R\Lambda}-H_0c_R)+c_{D1\Lambda}(H_0\tau_{R\Lambda}-H_0c_R)^5+c_{D2\Lambda} \\

	-\frac{1}{8}\left(\frac{23}{8}+5\ln4\right)=\frac{3}{2}(H_0\tau_{R\Lambda}-H_0c_R)^2-2c_{A1\Lambda}(H_0\tau_{R\Lambda}-H_0c_R) \\
	\quad -c_{A2\Lambda}+5c_{D1\Lambda}(H_0\tau_{R\Lambda}-H_0c_R)^4
\end{cases}
\end{equation}%
For the evolution of $\langle B\rangle$ we have%
\begin{equation}
	H_0^3u_B(\tau)=\begin{cases}
		\frac{1}{8}(H_0\tau)^2\left[\ln\left(\frac{H_0\tau}{4a_{RM}}\right)-\frac{17}{8}\right] & \tau\in[0;\tau_{R\Lambda}] \\

		-2c_{A1\Lambda}(H_0\tau-H_0c_R)^3 \\
		\quad-2c_{A2\Lambda}(H_0\tau-H_0c_R)^2\ln|H_0\tau-H_0c_R| \\
		\quad +\frac{1}{4}c_{D1\Lambda}(H_0\tau-H_0c_R)^6-c_{D2\Lambda}(H_0\tau-H_0c_R) \\
		\quad +c_{B\Lambda}(H_0\tau-H_0c_R)^2 & \tau\in[\tau_{R\Lambda};\tau(t_0)]
	\end{cases}
\end{equation}%
where the continuity fixes the integration constant $c_{B\Lambda}$ s.t.%
\begin{align}
	-\frac{1}{8}\left(\frac{17}{8}+\ln4\right)(H_0\tau_{R\Lambda})^2&=-2c_{A1\Lambda}(H_0\tau_{R\Lambda}-H_0c_R)^{3} \cr
	&-2c_{A2\Lambda}(H_0\tau_{R\Lambda}-H_0c_R)^{2}\ln|H_0\tau_{R\Lambda}-H_0c_R| \cr
	&+\frac{1}{4}c_{D1\Lambda}(H_0\tau_{R\Lambda}-H_0c_R)^{6}-c_{D2\Lambda}(H_0\tau_{R\Lambda}-H_0c_R) \cr
	&+c_{B\Lambda}(H_0\tau_{R\Lambda}-H_0c_R)^{2}
\end{align}

All perturbations are evaluated today, when it dominates the dark energy, as in
\S7.2. Substituting (\ref{A0 de}), (\ref{B0 de}) and (\ref{C0 de}) inside (\ref{sumract}), we obtain \emph{ract} and \emph{sum}
as well.

\subsection{No dark energy epoch}

A last possibility is that there are only the radiation and matter epochs. Applying (\ref{u rad}),
(\ref{u mat}) and (\ref{u de}), we get the same evolutions of $\langle A\rangle$, $\langle B\rangle$ and $\langle C\rangle$ as in \S7.3, just without the last
parts. All perturbations are evaluated today, when the matter dominates.%
\begin{align*}
	\langle A\rangle&=\frac{3}{2}\Omega_{IM0}[-\frac{1}{30}(2a^{1/2})^2+(2a^{1/2})^{-\frac{3}{2}}\cdot \cr
	&\cdot\left[c_{A1M}\sin\left(\frac{\sqrt{71}}{2}\ln(2a^{1/2})\right)+c_{A2M}\cos\left(\frac{\sqrt{71}}{2}\ln(2a^{1/2})\right)\right]] \Rightarrow
\end{align*}
\begin{align}
\label{A0 mat}
	\langle A\rangle_0&=\frac{3}{2}\Omega_{IM0}[-\frac{2}{15}+\frac{c_{A1M}}{2\sqrt{2}}\sin\left(\frac{\sqrt{71}}{2}\ln2\right)+\frac{c_{A2M}}{2\sqrt{2}}\cos\left(\frac{\sqrt{71}}{2}\ln2\right)], \cr
	\langle A\rangle'_0&=\frac{3}{2}\Omega_{IM0}[-\frac{2}{15}-\frac{3c_{A1M}+\sqrt{71}}{8\sqrt{2}}\sin\left(\frac{\sqrt{71}}{2}\ln2\right) \cr
	&-\frac{3c_{A2M}-\sqrt{71}}{8\sqrt{2}}\cos\left(\frac{\sqrt{71}}{2}\ln2\right)];
\end{align}

\begin{align*}
	\langle B\rangle&=\frac{3}{2}\frac{\Omega_{IM0}}{H_0}[-\frac{2}{175}(2a^{1/2})^3+\frac{1}{2}c_{D1M}(2a^{1/2})^{-2} \cr
	&+\frac{1}{5}c_{D2M}(2a^{1/2})-\frac{1}{50}(2a^{1/2})^{-\frac{1}{2}}[(3c_{A1M}+\sqrt{71}c_{A2M})\sin\left(\frac{\sqrt{71}}{2}\ln(2a^{1/2})\right) \cr
	&+(3c_{A2M}-\sqrt{71}c_{A1M})\cos\left(\frac{\sqrt{71}}{2}\ln(2a^{1/2})\right)]+c_{BM}(2a^{1/2})^{-4}] \Rightarrow
\end{align*}
\begin{align}
\label{B0 mat}
	H_0\langle B\rangle_0&=\frac{3}{2}\Omega_{IM0}[-\frac{16}{175}+\frac{1}{8}c_{D1M}+\frac{2}{5}c_{D2M}-\frac{3c_{A1M}+\sqrt{71}c_{A2M}}{50\sqrt{2}}\sin\left(\frac{\sqrt{71}}{2}\ln2\right) \cr
	&-\frac{3c_{A2M}-\sqrt{71}c_{A1M}}{50\sqrt{2}}\cos\left(\frac{\sqrt{71}}{2}\ln2\right)+\frac{1}{16}c_{BM}];
\end{align}

\begin{align*}
	\langle C\rangle&=\frac{3}{2}\Omega_{IM0}[-\frac{17}{150}(2a^{1/2})^2+c_{D1M}(2a^{1/2})^{-3}+c_{D2M}+ \cr
	&+\frac{1}{5}(2a^{1/2})^{-\frac{3}{2}}\left[c_{A1M}\sin\left(\frac{\sqrt{71}}{2}\ln(2a^{1/2})\right)+c_{A2M}\cos\left(\frac{\sqrt{71}}{2}\ln(2a^{1/2})\right)\right]] \Rightarrow
\end{align*}
\begin{align}
\label{C0 mat}
	\langle C\rangle_0&=\frac{3}{2}\Omega_{IM0}[-\frac{34}{75}+\frac{1}{8}c_{D1M}+c_{D2M}+\frac{c_{A1M}}{10\sqrt{2}}\sin\left(\frac{\sqrt{71}}{2}\ln2\right) \cr
	&+\frac{c_{A2M}}{10\sqrt{2}}\cos\left(\frac{\sqrt{71}}{2}\ln2\right)], \cr
	\langle C\rangle'_0&=\frac{3}{2}\Omega_{IM0}[-\frac{34}{75}-\frac{3}{16}c_{D1M}-\frac{3c_{A1M}+\sqrt{71}}{40\sqrt{2}}\sin\left(\frac{\sqrt{71}}{2}\ln2\right) \cr
	&-\frac{3c_{A2M}-\sqrt{71}}{40\sqrt{2}}\cos\left(\frac{\sqrt{71}}{2}\ln2\right)], \cr
	\langle C\rangle''_0&=\frac{3}{2}\Omega_{IM0}[\frac{15}{32}c_{D1M}+\frac{21c_{A1M}+12c_{A2M}+3\sqrt{71}}{160\sqrt{2}}\sin\left(\frac{\sqrt{71}}{2}\ln2\right) \cr
	&+\frac{21c_{A2M}-12c_{A1M}-5\sqrt{71}}{160\sqrt{2}}\cos\left(\frac{\sqrt{71}}{2}\ln2\right)].
\end{align}%
Substituting (\ref{A0 mat}), (\ref{B0 mat}) and (\ref{C0 mat}) inside (\ref{sumract}), we obtain \emph{ract} and \emph{sum}.

\section{Evaluations of dark matter and dark energy}

Now we employ the most recent measures of the cosmological parameters. The space
flatness is confirmed by \cite{Spergel:2003cb}%
\begin{equation}
	\Omega_{tot}=1-\Omega_{k0}=1.02\pm0.02.
\end{equation}%
Thus we can assume $\Omega_{R0}+\Omega_{M0}+\Omega_{\Lambda0}=\Omega_{tot}:=1$. We have also
\begin{equation}
	\quad \Omega_{R0}=8.24\cdot10^{-5}\pm10^{-7} \cite{Spergel:2003cb}, \quad \Omega_{M0}=0.315\pm0.007 \cite{measures}.
\end{equation}%
These are coherent with
\begin{equation}
	\Omega_{\Lambda0}=0.685\pm0.007, q_0=-0.527\pm0.0105.
\end{equation}%
Moreover,%
\begin{equation}
	\Omega_{BM0}=0.0486\pm0.0010 \cite{measures} \Rightarrow \Omega_{DM0}=0.266\pm0.008.
\end{equation}%
The fraction of matter unexplained by the CSM is $84.57\%$.

\subsection{Searching for good solutions}

For any value of the free parameter $\bar{\Omega}_{M0}$, we can get a numerical solution of $\bar{\Omega}_{R0}$. Following
\S7.2 we have, for any chosen value, the evolution of $a(\tau)$, and from the formulas of \S7.3,
7.4 or 7.5 the quantities \emph{ract} and \emph{sum}, and thus $\Omega_{IM0}$, $\Omega_{FM0}$ and $\Omega_{F\Lambda0}$.

Imposing (\ref{restr}), there could be one or more solutions for $\bar{\Omega}_{R0}$, or no one, depending on
$\bar{\Omega}_{M0}$. For any solution, we have to check if it is acceptable. The selfconsistence checking
will require to compute the evolution of $\langle\Psi\rangle$, since we have to find that its maximum is
less than $\frac{1}{2}$. Getting a set of selfconsistent and acceptable solutions, we will seek if some
of them are also good.

Applying this planwork with a numerical algorithm, we find that for a generic $\bar{\Omega}_{M0}$
there are up to two acceptable values of $\bar{\Omega}_{R0}$. E.g. we can find%
\begin{equation}
	\bar{\Omega}_{R0}|_{\bar{\Omega}_{M0}=0.5}\cong\begin{cases}
		10.97\cdot10^{-5} \\

		0.2216
	\end{cases}.
\end{equation}%
The set of solutions with $\bar{\Omega}_{R0}\sim10^{-4}$ have a radiation density quite near to the value of
the CSM. We can call them the \lq\lq principal\rq\rq~solutions, and \lq\lq secondary\rq\rq~solutions the others.
Indeed, following these solutions with continuity, for $\bar{\Omega}_{M0}=0.315\cong\Omega_{M0}$ we find trivially
\begin{equation}
	\Omega_{IM0}=0 \Rightarrow \Omega_{FM0}=\Omega_{F\Lambda0}=0, \bar{\Omega}_{w0}\equiv\Omega_{w0}.
\end{equation}%
The secondary solutions are not selfconsistent, since all of them have $\Omega_{IM0}>99\%\cdot\Omega_{TM0}$,
so that they break the Cosmological Principle. Moreover, the secondary solutions have
quite big perturbations $2\max\langle\Psi\rangle>0.5$: they are smaller than $1$ anyway, but not small enough.

On the other hand, the principal solutions are selfconsistent. $\frac{\Omega_{IM0}}{\Omega_{TM0}}$ becomes greater as
$\bar{\Omega}_{M0}$ runs away from $\Omega_{M0}$, but it is always less than $45\%$. It is the same for $2\max\langle\Psi\rangle$, which
is always smaller than $0.28\ll1$. \\
However, most of the principal solutions are not good. We find just a little interval around
$\bar{\Omega}_{M0}\cong0.2$ for which are explained some fraction of both dark matter and dark energy.
For the values $\bar{\Omega}_{M0}=0.2, \bar{\Omega}_{R0}\cong10.03\cdot10^{-5}, \bar{\Omega}_{\Lambda0}\cong0.8$ it is%
\begin{align}
	&\Omega_{IM0}\cong-0.0316, \quad \Omega_{TM0}\cong0.1327 \Rightarrow |\Omega_{IM0}|\cong23.8\%\cdot\Omega_{TM0}\ll\Omega_{TM0}, \cr
	&2\max\langle\Psi\rangle\cong0.0485\ll1, \cr
	&\Omega_{FM}\cong0.1823 \Rightarrow \Omega_{TDM0}\cong0.0841\cong63.39\%\cdot\Omega_{TM0}, \cr
	&\Omega_{F\Lambda}\cong0.0276 \Rightarrow \Omega_{T\Lambda0}\cong0.6572\cong95.96\%\cdot\Omega_{\Lambda0}.
\end{align}%

\begin{figure}[ht]
	\begin{minipage}[b]{1.0\linewidth}
		\centering
		\includegraphics[width=\textwidth]{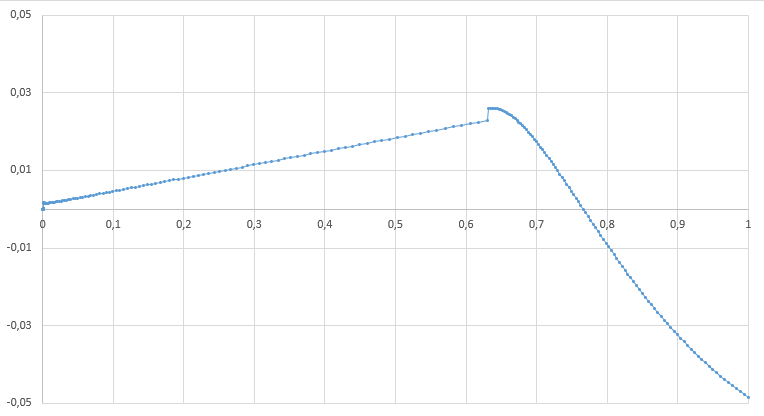}
		\caption{The evolution of the metric perturbation $2\langle\Psi\rangle(a)$, for the case $\bar{\Omega}_{M0}=0.2$.}
	\end{minipage}
\end{figure}
\clearpage

For $\bar{\Omega}_{M0}>0.2$ we start soon to have $\Omega_{FM0}<0$, so that the solutions are no more
good. For $\bar{\Omega}_{M0}<0.2$ it starts \emph{vice versa} to be $\Omega_{F\Lambda0}<0$, and the solutions are no more
good as well.

\subsection{Searching for solutions without dark energy or dark matter}

We can seek if there is a selfconsistent and acceptable solution which fully explains the
dark energy as fictitious. From the last paragraph, we know it would require an high $\bar{\Omega}_{M0}$,
for which $\Omega_{FM0}<0$ and the dark matter is more than in the CSM. \\
The condition of nonexistence of dark energy is $\bar{\Omega}_{\Lambda0}:=0$, to that it is automatically fixed
$\bar{\Omega}_{R0}=1-\bar{\Omega}_{M0}$. The (\ref{restr}) are solved by%
\begin{equation}
	\bar{\Omega}_{R0}\cong26.31\cdot10^{-5}, \quad \bar{\Omega}_{M0}\cong0.9997, \quad \bar{\Omega}_{\Lambda0}=0.
\end{equation}%
In such a case we find%
\begin{align}
	&\Omega_{IM0}\cong0.2516, \quad \Omega_{TM0}\cong0.5647 \Rightarrow |\Omega_{IM0}|\cong44.56\%\cdot\Omega_{TM0}\ll\Omega_{TM0}, \cr
	&2\max\langle\Psi\rangle\cong0.2792\ll1, \cr
	&\Omega_{FM}\cong-0.2507 \Rightarrow \Omega_{TDM0}\cong0.5161\cong91.39\%\cdot\Omega_{TM0}, \cr
	&\Omega_{F\Lambda}\cong0.685 \Rightarrow \Omega_{T\Lambda0}\cong0.
\end{align}

\begin{figure}[ht]
	\begin{minipage}[b]{1.0\linewidth}
		\centering
		\includegraphics[width=\textwidth]{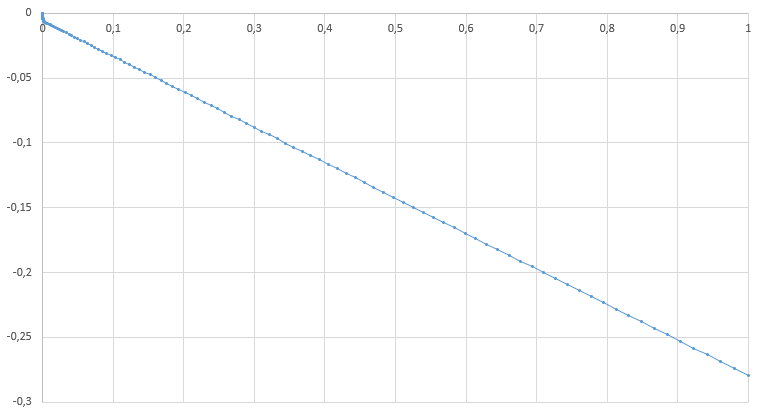}
		\caption{The evolution of the metric perturbation $2\langle\Psi\rangle(a)$, for the case fully explaining
the dark energy.}
	\end{minipage}
\end{figure}
\clearpage

On the opposite, we can seek if there is a selfconsistent and acceptable solution that fully explains
the dark matter as fictitious. From the last paragraph, we know it would require a small
$\bar{\Omega}_{M0}$, for which $\Omega_{F\Lambda0}<0$ and the dark energy is more than in the CSM. \\
The condition of nonexistence of dark matter is $\bar{\Omega}_{TM0}:=\Omega_{MB0}$. The corresponding value
of $\bar{\Omega}_{R0}$ is fixed by (\ref{restr}), which we solve numerically%
\begin{equation}
	\bar{\Omega}_{R0}\cong9.59\cdot10^{-5}, \quad \bar{\Omega}_{M0}\cong0.0819, \quad \bar{\Omega}_{\Lambda0}\cong0.9170.
\end{equation}%
In such a case we find%
\begin{align}
	&\Omega_{IM0}\cong-0.0218, \quad \Omega_{TM0}=\Omega_{MB0}\cong0.0486 \Rightarrow |\Omega_{IM0}|\cong44.84\%\cdot\Omega_{TM0}\ll\Omega_{TM0}, \cr
	&2\max\langle\Psi\rangle\cong0.1608\ll1, \cr
	&\Omega_{FM}=\Omega_{DM0}\cong0.266 \Rightarrow \Omega_{TDM0}=0, \cr
	&\Omega_{F\Lambda}\cong-0.1039 \Rightarrow \Omega_{T\Lambda0}\cong0.7888\cong115.17\%\Omega_{\Lambda0}.
\end{align}

\begin{figure}[ht]
	\begin{minipage}[b]{1.0\linewidth}
		\centering
		\includegraphics[width=\textwidth]{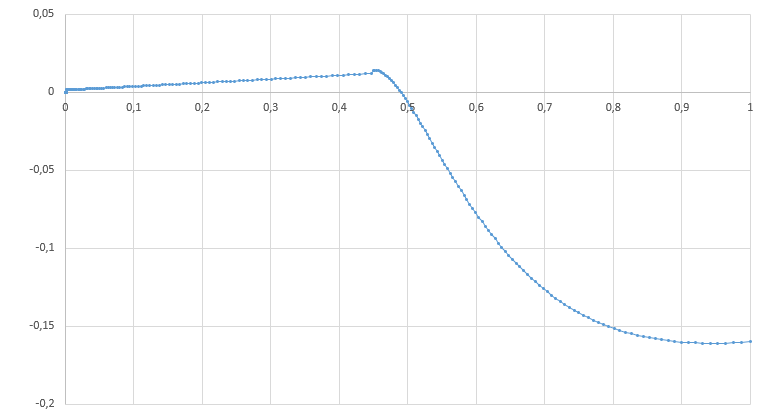}
		\caption{The evolution of the metric perturbation $2\langle\Psi\rangle(a)$, for the case fully explaining
the dark matter.}
	\end{minipage}
\end{figure}
\clearpage

\subsection{The principal solutions}

For values $\bar{\Omega}_{M0}>0.9997$ we would find $\Omega_{T\Lambda0}<0$, which is not acceptable. For values
$\bar{\Omega}_{M0}<0.0819$ we would find $\Omega_{TM0}<\Omega_{BM0}$, which is not acceptable. This mean that the
acceptable principal solutions range in the interval $\bar{\Omega}_{M0}\in[0.0819; 0.9997]$.

We summarize with the following graphics the numerical results about the principal
solutions.

\begin{figure}[ht]
	\begin{minipage}[b]{1.0\linewidth}
		\centering
		\includegraphics[width=\textwidth]{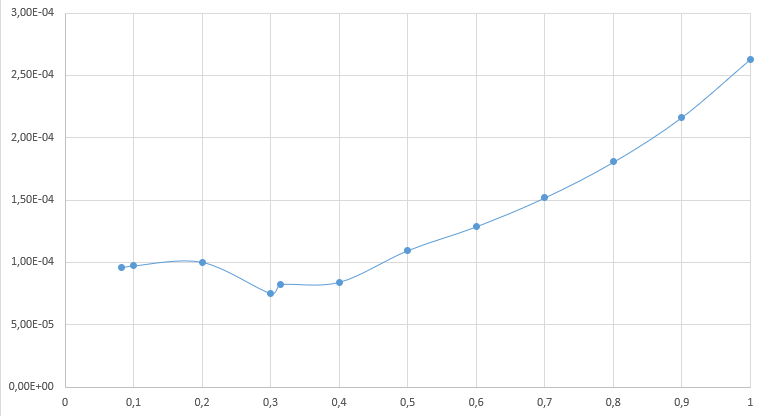}
		\caption{Background radiation density $\bar{\Omega}_{R0}$, depending on the parameter $\bar{\Omega}_{M0}$.}
	\end{minipage}
\end{figure}

\begin{figure}[ht]
	\begin{minipage}[b]{1.0\linewidth}
		\centering
		\includegraphics[width=\textwidth]{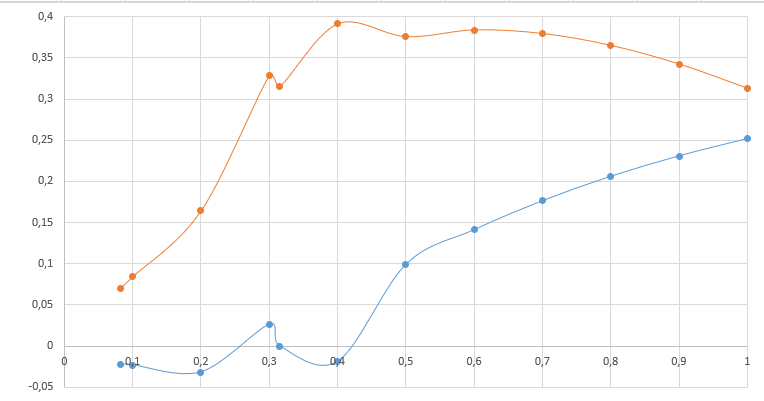}
		\caption{Inhomogeneous matter $\Omega_{IM0}$ (blue line) and homogeneous matter $\Omega_{HM0}$ (red
line), depending on the parameter $\bar{\Omega}_{M0}$.}
	\end{minipage}
\end{figure}

\begin{figure}[ht]
	\begin{minipage}[b]{1.0\linewidth}
		\centering
		\includegraphics[width=\textwidth]{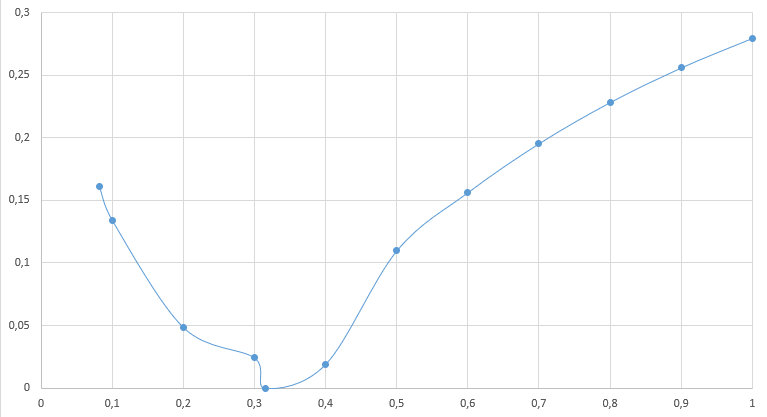}
		\caption{Maximum of the metric perturbation $2\max\langle\Psi\rangle$, depending on the parameter
$\bar{\Omega}_{M0}$.}
	\end{minipage}
\end{figure}

\begin{figure}[ht]
	\begin{minipage}[b]{1.0\linewidth}
		\centering
		\includegraphics[width=\textwidth]{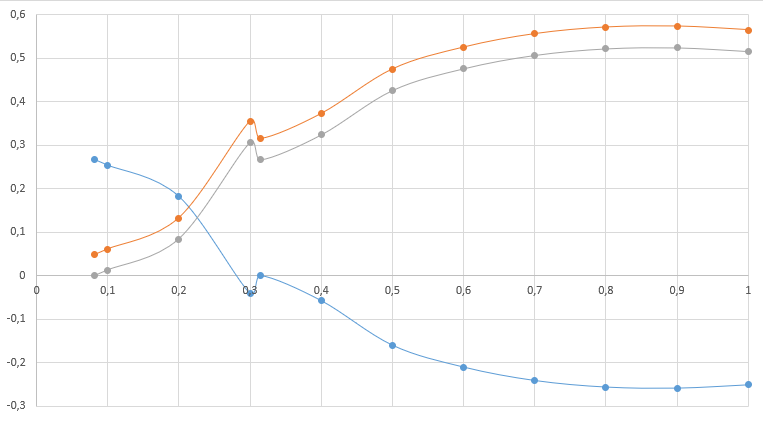}
		\caption{Fictitious matter $\Omega_{FM0}$ (blue line), true matter $\Omega_{TM0}$ (red line) and true dark matter
$\Omega_{TDM0}$ (grey line), depending on the parameter $\bar{\Omega}_{M0}$.}
	\end{minipage}
\end{figure}

\begin{figure}[ht]
	\begin{minipage}[b]{1.0\linewidth}
		\centering
		\includegraphics[width=\textwidth]{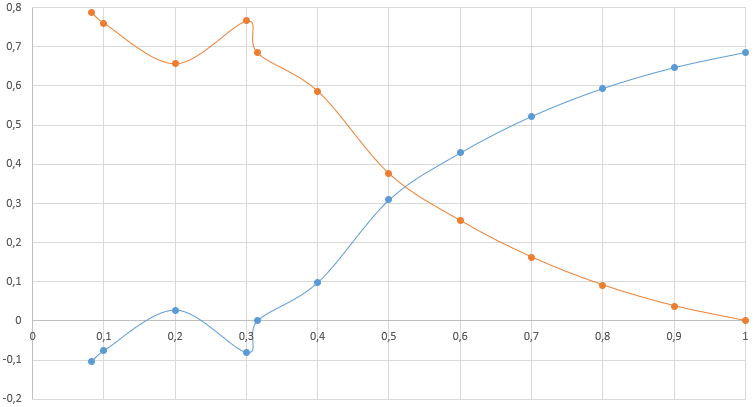}
		\caption{Fictitious dark energy $\Omega_{F\Lambda0}$ (blue line) and true dark energy $\Omega_{T\Lambda0}$ (red line),
depending on the parameter $\bar{\Omega}_{M0}$.}
	\end{minipage}
\end{figure}

\clearpage
\section{Conclusions and future perspective}

In this paper we developed the framework of \cite{re}, managing to apply it to a model of
our universe, complete with all the components. The large number of variables leaves a
free parameter, depending on which we found a one-dimensional set of possible solutions.
Within this interval, more dark matter is explained less as less dark energy is, and \emph{vice versa}.
At an end of the range, dark matter is fully explained as a relativistic effect, but the
same effect caused an underestimation of dark energy in the Cosmological Standard
Model. At the other end the numbers are analogous, with dark matter and dark energy
exchanged. For a particular value, both dark matter and dark energy found a partial
explanation. An additional measure would determine which is the right solution, but
anyway a correction of the parameters of the CSM is required.

Better measures of the density parameters will improve our estimations of $\Omega_{FM0}$ and
$\Omega_{F\Lambda0}$, but they can't fix the right parameter $\bar{\Omega}_{M0}$. The difference between $\bar{\Omega}_{R0}$ and the
measured $\Omega_{R0}$, e.g., is not matter of measure precision, but of the factor $\left(\frac{H_0}{\bold{H}_0}\tilde{a}_0\right)^2$, which
concerns the background universe and is not measurable. Rather, a measure of the actual
gravitational force $\vec{\nabla}\Psi(\underline{x};t_0)$ could put the restraint we need. Another possible measure
could be the estimation of the matter inhomogeneity at large scale $\Omega_{IM0}$, i.e. the deviation
from the exact Cosmological Principle.

We approximated our calculations in many points. To overcome them would be an
improvement of the framework. Solving numerically the evolution $a(\tau)$ it would not be
necessary any sticking, which presumably would return more regular graphics than in \S8.3;
but recall that this would require the form of $T(\tau)$ for the multi-component case.
Even if we found $2\max\langle\Psi\rangle$ always far smaller than $1$, it could not be considered fully negligible,
so that an higher order calculation could provide some relevant corrections. Moreover, we
assumed a spatially flat background metric and an irrotational matter, which is not
the most general framework.

In the present article we considered the global dark matter effects only. Our cosmological
model requires also the calculation of the local effects, to be empirically verified.
This needs to overcome the averaging of $g_{\mu\nu}$, and the distribution of fictitious dark matter
would depend on the spatial distribution of inhomogeneities $\tilde{\rho}(\underline{x};t_0)$. A study of such
distribution could start from the fractal properties of the matter structures at large scales
\cite{fract}, \cite{fract2}. The fluctuations of the resultant potential $\Psi(\underline{x};t_0)$ should be compared to the
dark matter halo of the galaxies.

The study of local dark matter effects would provide corrections to the standard newtonian
approximations for the dynamic of galaxies and clusters. For such calculations,
we cannot assume an irrotational matter as we did here. The rotation of galaxies could
provide a rotational term for the non diagonal components of the metric $\vec{B}$, which contributes
to fictitious gravitational effects \cite{Balasin:2006cg}, \cite{Crosta:2018var}. Finally, the total amount of the local
fictitious effects could be compared to the global fictitious effect $\Omega_{FM0}$ we found here, and
the equivalence between them could be the additional restraint we need to fix uniquely
the parameter $\bar{\Omega}_{M0}$.

\medskip
\textbf{Acknowledgments}. We thank Oliver Piattella, Mariateresa Crosta, Marco Giammaria,
Alessio Marrani and Alessio Notari for useful discussions.

\end{document}